\documentclass[11pt]{article}
\usepackage[top=1in, bottom=1in, left=1in, right=1in]{geometry}
\usepackage{booktabs} 
\usepackage[ruled]{algorithm2e} 

\SetAlFnt{\small}
\SetAlCapFnt{\small}
\SetAlCapNameFnt{\small}
\SetAlCapHSkip{0pt}
\IncMargin{-\parindent}

\usepackage{amsmath}
\usepackage{amsthm}
\usepackage{amsfonts}
\usepackage{amssymb}
\usepackage{algpseudocode}
\usepackage{mathtools}
\usepackage{bm}
\usepackage{tikz}
\usetikzlibrary{positioning}

\newtheorem{lemma}{Lemma}
\newtheorem{theorem}{Theorem}
\newtheorem{prop}{Proposition}

\newtheorem{cor}{Corollary}
\newtheorem{definition}{Definition}

\renewcommand{\bar}{\overline}
\renewcommand{\hat}{\widehat}
\renewcommand{\tilde}{\widetilde}

\DeclareMathOperator*{\argmax}{arg\,max}

\newcommand{\reals}{\mathbb{R}}

\newcommand{\ab}{\alpha}
\newcommand{\bmax}{\Delta^{max}_{b}}
\newcommand{\gammai}{\gamma_i}
\newcommand{\gammaj}{\gamma_j}

\newcommand{\udi}{u^e_i}
\newcommand{\udihat}{\bar{u}^e_{i}}

\newcommand{\Udi}{U^e_i}

\newcommand{\ehat}{\hat{e}}
\newcommand{\Ehat}{\hat{E}}

\newcommand{\cit}{c_t^{i}}

\newcommand{\Oh}{\mathcal{O}}

\newcommand{\T}{\mathcal{T}}

\newcommand{\Es}{\mathcal{E}}
\newcommand{\I}{\mathcal{I}}
\DeclareMathOperator{\E}{\mathbb{E}}
\DeclareMathOperator{\Pb}{\mathbb{P}}

\algdef{SE}[SUBALG]{Indent}{EndIndent}{}{\algorithmicend\ }%
\algtext*{Indent}
\algtext*{EndIndent}

\usepackage[numbers]{natbib}
\setcitestyle{number}
\newcount\Comments  
\usepackage{color} 
\newcommand{\kibitz}[2]{\ifnum\Comments=1\textcolor{#1}{#2}\fi}

\newcommand{\squishlist}{
   \begin{list}{$\bullet$}
    { \setlength{\itemsep}{0pt}      \setlength{\parsep}{3pt}
      \setlength{\topsep}{3pt}       \setlength{\partopsep}{0pt}
      \setlength{\leftmargin}{1.5em} \setlength{\labelwidth}{1em}
      \setlength{\labelsep}{0.5em} } }
\newcommand{\squishend}{  \end{list}  }

\begin{document}
\title{Equilibrium of Data Markets with Externality}
\author{        
        Safwan Hossain \\ 
        Harvard University \\ 
        \texttt{shossain@g.harvard.edu} 
        \and 
        Yiling Chen \\ 
        Harvard University \\
        \texttt{yiling@seas.harvard.edu}
}
\date{}
\maketitle

\begin{abstract}
    We model real-world data markets, where sellers post fixed prices and buyers are free to purchase from any set of sellers, as a simultaneous game. A key component here is the negative externality buyers induce on one another due to data purchases. Starting with a simple setting where buyers know their valuations a priori, we characterize both the existence and welfare properties of the pure Nash equilibrium in the presence of such externality. While the outcomes are bleak without any intervention, mirroring the limitations of current data markets, we prove that for a standard class of externality functions, platforms intervening through a transaction cost can lead to a pure equilibrium with strong welfare guarantees. We next consider a more realistic setting where buyers learn their valuations over time through market interactions. Our intervention is feasible here as well, and we consider learning algorithms to achieve low regret concerning both individual and cumulative utility metrics. Lastly, we analyze the promises of this intervention under a much richer externality model.
\end{abstract}

\section{Introduction}
Data plays a central role in machine learning, and demand for it has grown substantially due to the increasing value it provides. It has thus become the subject of trading, and understanding markets for data (or more generally, information) has gained traction in research communities in recent years \citep{bergemann2018design, babaioff2012optimal, bergemann2019markets, mehta2019sell, agarwal2019marketplace, agarwal2020towards}. Such works have largely focused on determining the ``right'' price for data so that it is allocated to those who value it most \citep{bergemann2019markets,mehta2019sell,agarwal2019marketplace, agarwal2020towards}.
The mechanisms proposed here for two-sided markets are usually auction-based, and 
require buyers to report their valuations for data sources \citep{agarwal2019marketplace, agarwal2020towards}. These requirements are challenging, if not impossible, to achieve in practice due to the peculiar characteristics of data products. Buyers are generally free to purchase from any combination of data sources, and it is impractical to elicit accurate valuations in this combinatorially large space. More importantly, a data source's value to a buyer largely arises from how it can improve model performance, a metric the buyer herself may not know before making purchases and evaluating the results. Thus, she acts in the market with only partial or incomplete information. 
It is hence not surprising that most real-world data market platforms, such as Snowflake Data Marketplace\footnote{https://www.snowflake.com/} and AWS Data Exchange\footnote{https://docs.aws.amazon.com/data-exchange/index.html}, take a very simple format: sellers use some pricing mechanism to set a price for each dataset, and buyers can freely choose which datasets to purchase, with platforms charging a simple transaction fee. While sellers can change prices over time to reflect market changes, the granularity of this is quite coarse, and prices stay stable or \emph{fixed} for some time. The simplicity of fixed-price data markets is not just a huge operational advantage but also arguably necessary in light of the challenges data products pose. Such markets are the focus of our paper.

An important component of modeling data markets is accounting for the (usually negative) externality one buyer's purchase decision has on another. In a competitive setting, a buyer's value for data is predicated on the relative advantage it provides, with respect to their peers. Equivalently, the value of data can depend on what others have access to. While externality is not unique to data and is present for other products as well, it is especially prominent in data markets due to another salient feature: data can be replicated at a mass scale with zero marginal cost and sold to multiple buyers, a phenomenon that exacerbates any externality data products induce. Replication need not even be exact: sellers can offer different versions of the same dataset by injecting noise or by interleaving it with something innocuous \citep{pei2020survey}. Accordingly, real-world data markets usually give no guarantee of unique ownership. Data externality is thus a persistent phenomenon, and by linking one buyer's utility to another's decisions, it turns the buyer interactions into a game. We thus model fixed-price data markets between arbitrary buyers and sellers in the presence of externality and free replicability as a game. Despite their simplicity, which sidesteps many of the above concerns, this model has not been formally analyzed in literature and its properties are unknown. By leveraging the formalisms of game theory, we systematically explore this landscape and ask:
\begin{itemize}
    \item How well can fixed-price markets serve the purpose of allocating data to buyers, especially with respect to social welfare? 
    \item How does buyer externality affect their strategy and market performance?
    \item Can simple platform interventions improve the outcome of fixed-price data markets? 
    \item If buyers learn their value for data sources by repeatedly interacting with a data market, can they learn to act optimally, and how much social welfare can be achieved? 
\end{itemize}

\subsection{Our contributions} 
We model the interactions between buyers in fixed-price data markets as a simultaneous game, with buyer utility depending on the net gain derived from the purchased data, the externality they suffer due to others' actions, and possibly some cost administered by the platform. We formally define this in section \ref{section:model}. In section \ref{section:known_valuations}, we study the pure Nash equilibrium (PNE) of this game and show that without any intervention, this equilibrium has poor social welfare, confirming the present concerns with data markets. A market intervention in the form of a transaction cost is then proposed that significantly improves the situation by guaranteeing a dominant strategy NE whose social welfare approaches the optimal. Section \ref{section:learned_valuations} then considers the real-world setting wherein buyers do not know valuations a priori but instead must learn them through repeated market interactions. We model the buyer's learning problem as a multi-armed bandit instance with exponentially many arms, show the proposed platform intervention is feasible here as well, and consider online learning algorithms for buyers in this challenging combinatorial setting. Section \ref{section:rich_externality} explores a richer model of buyer externality, wherein not intervening does not guarantee even a reasonable approximate equilibrium. In contrast, our proposed intervention significantly improves this, guaranteeing a $\varepsilon-$pure NE with good welfare. Nonetheless, this is a more challenging model for the learning setting, and we discuss these along with broader extensions in section \ref{sec:discussion}.

\section{Related Works}\label{section:related_works}



\citet{bergemann2019markets} provide a broad overview of the growing literature on markets for information and data. Works here examine how to optimally sell information according to some objectives, and take a mechanism design approach, with a focus on the seller's pricing problem. For instance, many works consider a monopoly information holder directly selling a private random signal to buyers \citep{esHo2007price,bergemann2018design,babaioff2012optimal,chen2020selling,cai2020sell,bonatti2022selling}. The seller decides on a menu of information products (e.g. partially revealing the signal) and an associated price for each product to maximize her profit. \citet{mehta2019sell} also study selling a dataset in a monopoly seller setting. When there are multiple sellers and multiple buyers, auction-based mechanisms have also been leveraged to design two-sided data marketplaces \citep{agarwal2019marketplace, agarwal2020towards}. In such works, the marketplaces themselves are not profit-driven but intend to maximally facilitate the matching of data to buyers. \citet{bergemann2019markets} and \citet{BergemannSocialData2022} also study profit-driven data intermediaries who make a bilateral deal to purchase information from data holders and then sell the information to data buyers. Here, the data holders are consumers, and the buyers are firms that use purchased information to price discriminate.  

While determining how to price data is fundamentally important, the proposed mechanisms so far have not found their way into real-world data markets. A menu of partially-revealing information products can be too complex and cumbersome for buyers. As discussed in the introduction, buyers are unlikely to know their valuation for data a priori, rendering mechanisms that solicit buyer valuations (e.g. auctions) impractical. Our work thus takes an orthogonal direction and steps away from the design question of how to sell data and consider data sellers participating in a real-world fixed-price data market. We take this market mechanism as a given and adopt a game-theoretic approach to understanding buyer behavior and dynamics in such markets. 

A key feature of data markets that is central to our model is buyer externality. This has been highlighted in studying monopoly data selling \citep{ADMATI1986400, Admati1990, Bimpikis2019, bonatti2022selling} and data auctions \citep{agarwal2020towards} in competitive environments. These works explicitly model downstream competition (e.g. trading in financial markets) among data buyers and the negative buyer externality that arises therein. Our work abstracts away the specifics of competitive environments to present such externality in a general sense. Another type of externality present in information markets is the externality among data sellers. Naturally, the value of one's data decreases when others decide to sell similar data \citep{BergemannSocialData2022}. Seller externality is an important phenomenon to consider, which we leave for future work.

Lastly, our work on the online setting relaxes the assumption that buyers know their valuations a priori. This spiritually parallels studying learning agents in other mechanisms such as auctions \cite{Blum2003,DBLP:journals/corr/WeedPR15,braverman_selling_2018} and peer prediction \cite{feng2022peer}. The techniques used in the online section leverage the vast literature on bandits, particularly bandits in metric spaces \citep{kleinberg2019bandits, bubeck2008online}. We view studying markets with learning agents as a step toward a more realistic evaluation of market performance.   

\section{Model}\label{section:model}

\paragraph{\textbf{Market Structure:}} Consider $n$ buyers, $\mathcal{D} = \{d_1,$ $\dots, d_n\}$, and $k$ sellers, $\mathcal{L} = \{\ell_1, \dots \ell_k \}$, who without loss of generality have one dataset each to sell. We consider a market where buyers are free to purchase from any subset of the sellers, each of whom posts a fixed price for their dataset and cannot refuse to sell. The data offered by the sellers can be arbitrarily correlated, and multiple buyers may buy from the same seller. Let $\Gamma$ denote the power set of the sellers, and let $\gamma \in \Gamma$ denote a specific subset of sellers (we often refer to it as a \emph{seller set} or \emph{order}). We use subscripts to distinguish between seller sets chosen by various buyers - for example, $\gammai$ refers to the set of buyers chosen by buyer $i$ - and superscripts to distinguish or index between two seller sets irrespective of buyers - for example, $\gamma^m$ and $\gamma^n$ refers to two distinct seller sets. We consider buyers simultaneously submitting their orders and use $S = (\gamma_1, \dots, \gamma_n)$ to denote the purchase orders of all buyers.  


\paragraph{\textbf{Buyer's Utility:}} The utility of buyer $i$ when purchasing from a seller set $\gamma$ is predicated on: (1) the increased value/performance due to purchased data, denoted by the random variable $P_i(\gamma)$ with $\E[P_i(\gamma)] \triangleq p_i(\gamma)$, (2) the cost charged by sellers in $\gamma$, $c(\gamma)$, (3) the negative externality caused by another buyer $j$'s action, denoted by the random variable $E_{ij}(\cdot)$ with $\E[E_{ij}(\cdot)] \triangleq e_{ij}(\cdot)$, and (4) a \emph{transaction cost} $\mathcal{T}_i(S)$ charged by the platform, which specifies the amount that buyer $i$ needs to pay (or receive) when the set of purchase orders is $S$. Without loss of generality, we assume $P_i(\gamma) \in [0,1]$ and $c(\gamma) \in [0,1]$, and the ``buy nothing'' option ($\gamma = \emptyset$) has 0 increased performance and 0 cost. Note a buyer that decides to ``buy nothing'' and not participate in the market, will still suffer the negative externality induced by others. We primarily focus on the most externality model in literature where buyer $i$'s externality from buyer $j$'s action depends arbitrarily on the latter's action: $E_{ij}(\gamma_j)$ \citep{aseff2008optimal, li2019facility, agarwal2020towards}. To ensure that the total externality caused by a buyer is at most $1$ and thus in the same range as performance and cost, we assume that $\forall \, i$, $\sum_{j \ne i}E_{ji}(\gammai) \in [0,1]$. In section \ref{section:rich_externality}, we consider a richer class of externality functions and in appendix \ref{appendix:externality}, we include a detailed discussion on the economic interpretations of these externality models. 
We now formally define the buyer's utility, and for brevity, combine the performance increase $P_i(\gamma)$ and seller costs $c(\gamma)$ terms into a single \emph{net gain} term $G_i(\gamma)$, with $\E[G_i(\gamma)] \triangleq g_i(\gammai)$.

\begin{definition}[Utility and Welfare]
    For seller cost function $c$, and buyer $i$ with performance increase $P_i$, we define her net gain to be $G_i(\gamma) = P_i(\gamma) - c(\gamma)$, with $G_i(\gamma) \in [-1,1]$. Using this, along with externality $E_{ij} \, \forall \, j$, and transaction cost $\mathcal{T}_i$, we define her stochastic utility for a complete order profile $S = (\gamma_1, \dots, \gamma_n)$ to be:
    \begin{equation*}
        U_i(S) =G_i(\gamma_i) - \sum_{j \ne i} E_{ij}(\gamma_j) - \mathcal{T}_i(S).
        \vspace{-1em}
    \end{equation*}
    The expected buyer utility is denoted by $\E[U_i(S)] \triangleq u_i(S)$. We define the expected social welfare of an order profile $S$ as the sum of all expected buyer utility: $sw(S) = \sum_{i}{u_i(S)}$.
\end{definition}



\paragraph{\textbf{Game and Solution Concept:}} 
We model buyers in the data market as playing a simultaneous-move game to maximize their expected utility. We note that agent $i$'s utility, and thus her best response, doesn't depend on any other agent's net gain $G_j(\cdot)$, which may be beneficial. We are primarily interested in analyzing the pure-strategy Nash equilibrium of this game under externality and transaction costs. Beyond the existence of such equilibria, we also aim to compare the social welfare at equilibrium to the optimal social welfare possible. A common notion for comparison is the \emph{price of stability}, the ratio of the optimal social welfare to the welfare of the best equilibrium \citep{nisan2007algorithmic, schulz2003performance}. However, this is a multiplicative metric and thus unamenable to additive notions of regret that are common in the online analysis we do in section \ref{section:learned_valuations}. As such, we define a comparable additive notion called \emph{welfare regret at equilibrium (WRaE)} to characterize the societal cost of self-interested behavior in games. 
\begin{definition}[WRaE]
    Let $S^* = \argmax_{S}{\sum_{i}{\bar{u}_i(S)}}$ be the optimal strategy with respect to social welfare and let $\mathcal{S}^q$ be the set of all equilibrium strategies. We define the \textbf{welfare regret at equilibrium} (WRaE) for our game as: $sw(S^*) - \max_{S \in \mathcal{S}^q}sw(S)$.
\end{definition}

\paragraph{\textbf{Known vs Unknown Utilities:}} 
Buyers may not know their expected net gain for a dataset until using it with their models. Thus, there are two sources of tension in the general setting: buyers behaving strategically with respect to their utility, and buyers learning these by interacting with the market. Although intertwined, the former relates to game dynamics and the latter is a learning theoretic question. To comprehensively study both, in section \ref{section:known_valuations} we start with the simpler case of buyers knowing the expected utility for their actions, which is standard in game theory \citep{prisner2014game, munoz2019pure, mailath1991extensive}, and focus on the resulting equilibrium and welfare properties. The more general online setting wherein buyers learn their valuations through repeated market participation and behave strategically based on these is studied next. Our key metrics here are online effective and online welfare regret (see section \ref{section:learned_valuations} for precise definition); these intuitively capture how well buyers can learn their optimal strategy, and how this learned strategy affects welfare regret at equilibrium.

\section{Data Markets Game with Known Utility}\label{section:known_valuations}
In this section, we consider the simultaneous-move data markets game wherein each buyer knows their expected net gain and externality. This allows us to settle game/market dynamics questions first, before considering the more general setting of buyers playing with learned valuations in section \ref{section:learned_valuations}. Considering this game without any platform intervention -- i.e. buyers simply pay the seller's cost and $\mathcal{T}_i(S) = 0$ -- admits a unique, dominant-strategy equilibrium: each buyer simply chooses the seller set with the highest net gain, since the externality suffered depends on the others' action. However, the welfare regret of this dominant-strategy equilibrium can be maximal. Intuitively, we sketch an instance where the seller set with the highest net gain also induces a high externality.

\begin{prop}\label{thrm:independent_externality_eq}
    For the data markets game with no intervention - $\mathcal{T}_i(S) = 0, \,\forall i$, it is a dominant strategy for any buyer $i$ to select $\gammai^d = \argmax_{\gamma}{g_i(\gamma)}$. However, there exists an instance of this game where the WRaE is maximal - (i.e. $\Theta(n)$ upper and lower bound).
\end{prop}
\begin{proof}
    Note each buyer $i$'s utility under these conditions is: $g_i(\gammai) - \sum_{j\ne i}{e_{ij}(\gammaj)}$. The only aspect of this utility that buyer $i$ can affect is $\gammai$, and thus she has a dominant strategy of choosing the source with the highest gain \footnote{Note that choosing the empty set and not participating is a valid strategy. Buyer $i$ suffers externality regardless of what she chooses}. It is intuitive that everyone adopting such a strategy will not always lead to good welfare. To see that there is an instance that achieves maximal WRaE, suppose the number of sellers and buyers are equal ($n = k$). For $k \in [1,\dots,n]$, define $\gamma^k$ as the seller set containing only seller $\ell_k$. For a buyer $i$, let $g_i(\gamma^i) = 1$ and $g_i(\gamma) = 1 - \epsilon$ for any $\gamma \ne \gamma^i$. Further, for any buyer pair $i, j$, let $e_{ij}(\gamma_j = \gamma^j) = \tfrac{1}{n-1}$ and $e_{ij}(\gamma_j=\gamma) = 0 \, \forall \gamma \ne \gamma^j$. In this instance, the unique dominant strategy/equilibrium is each buyer $i$ selecting $\gamma^i$, which results in 0 utility for all buyers. However, if each buyer $i$ selects any other seller set aside from $\gamma^i$, they achieve utility $1 - \epsilon$ each. Since this is the only equilibrium in this instance, the WRaE is $n - n\epsilon \rightarrow n$ as $\epsilon \rightarrow 0$. 
\end{proof}


 This result illustrates the very real and unsatisfactory phenomenon that occurs in modern data markets. Buyers, not incentivized to care about their impact on others, make myopic decisions based purely on their own net gain, to the detriment of both the individual and the collective. This impact is exacerbated by the easy replicability of data, allowing multiple buyers to purchase the same dataset. To address this, we consider market intervention in the form of a transaction cost $\mathcal{T}_i(S)$ levied by the platform, a standard approach in online marketplaces. We posit that platforms can use learning models to estimate externalities(Definition \ref{definition:predicted_externality}) and propose they charge each buyer a cost proportional to the net difference in estimated externality they are contributing to and suffering from (Definition \ref{definition:net_transaction_fee}).

\begin{definition}[Predicted Externality]\label{definition:predicted_externality}
     For any $i,j$, let $\Ehat_{ij}(\cdot)$ denote the predicted version of the externality $E_{ij}(\cdot)$, with $\E[\Ehat_{ij}(\cdot)] = \ehat_{ij}(\gamma)$ and $b_{ij}(\cdot) = e_{ij}(\cdot) - \ehat_{ij}(\cdot)$ representing the bias/error of the predictor. We assume predictions satisfy the same scale -- i.e. $\sum_{i}{\Ehat_{ji}(\cdot)} \in [0,1]$. As shorthand, we define $\bmax = \sum_{i \ne j}{\max_{\gamma^1, \gamma^2}|b_{ij}(\gamma^1) - b_{ij}(\gamma^2)|}$
\end{definition}

\begin{definition}[Transaction Cost]\label{definition:net_transaction_fee}
    For a platform-chosen parameter $\alpha \in [0,1]$, and a complete set of purchase orders $S$, we propose the platform charge each buyer $i$ a transaction cost $\mathcal{T}_i(S) = \alpha \left(\sum_{j \ne i}{ \Ehat_{ji}(\gammai)} - \sum_{j \ne i} \Ehat_{ij}(\gammaj)\right)$, with $\alpha = 0$ corresponding to no intervention.
\end{definition} 

Observe that we make no assumptions about the predictions $\hat{E}_{ij}$ and characterize it generally through their itemized bias or expected inaccuracy $b_{ij}$. Secondly, while the proposed cost $\T_i(S)$ can be negative, from the platform's perspective this is just a redistribution and not an actual payment, since the sum of the transaction cost is always 0 - i.e. revenue-neutral. This may be desirable since the addition of this cost does not diminish the cumulative social welfare of the buyers.  However, if platforms do wish to extract revenue they can charge an additional constant cost on top of this without changing any of our results. It is also advantageous that our proposed cost only depends on the externalities and not the net gain. We denote an \emph{instance} of our game with this intervention by the tuple $\I = (g_i, e_{ij}, b_{ij}, \alpha)$. Lastly, the proposed transaction cost renders each buyer $i$'s \emph{effective utility}, the component of their utility that depends on their action $\gammai$, to have the following form: $G_i(\gammai) - \ab \sum_{j}{\Ehat_{ji}(\gammai)} \triangleq \Udi(\gammai)$, with $\E[\Udi(\gamma)] \triangleq u_i^e(\gamma)$.

With this new cost in place, Theorem \ref{theorem:wrae_offline} gives a significantly improved result --- a dominant strategy still exists, with each buyer maximizing their expected effective utility, but WRaE is tightly bounded by $\Theta(n(1-\alpha))$, as the difference in bias parameters, $\bmax$ tend toward 0. This is a strong result since the upper bound implies that on \emph{all} instances, equilibrium has welfare regret less than $n(1-\alpha) + \bmax$. Crucially, this decreases linearly with $\alpha$, with $\alpha \rightarrow 1$ leading to the equilibrium solution attaining the optimal social welfare, up to prediction inaccuracies. The corresponding lower bound, which is independent of the bias terms, shows that this is essentially tight.

\begin{theorem}\label{theorem:wrae_offline}
    Under the proposed transaction cost, it is a dominant strategy for any buyer $i$ to select $\gammai^d = \argmax_{\gamma}{g_i(\gamma)- \alpha \sum_{j \ne i}{\ehat_{ji}(\gamma)}} = \argmax_{\gamma}{u_i^e(\gamma)}$. Further, WRaE is upper bounded by $n(1-\alpha) + \bmax$ and lower bounded by $n(1-\alpha)$.
\end{theorem}
\begin{proof}
    The expected utility under the given transaction cost is given by: $g_i(\gammai) - \sum_{j \ne i}e_{ij}(\gammaj) - \ab\left(\sum_{j \ne i}\ehat_{ji}(\gammai) - \sum_{j \ne i}{\ehat_{ij}(\gammaj)}\right)$. Since a buyer can only influence $\gammai$, the second and last terms can essentially be ignored, and the buyer has a dominant strategy, which is to maximize the expected effective utility, regardless of how others act. Regarding WRaE, we first prove the upper bound, before showing a specific instance achieves this upto bias terms.

    \emph{Upper bound:} For an instance $\I$, let $S^* = (\gamma_{1}^*, \gamma_{2}^*, \dots, \gamma_{n}^*)$ denote the social welfare optimal, and let $S^d = (\gamma^d_{1}, \gamma^d_{2}, \dots, \gamma^d_{n})$ denote the dominant strategy taken by the buyers. 
    The following is a direct implication of $S^d$ being the dominant strategy for all buyers: $ g_i(\gammai^d) - \alpha \sum_{j \ne i}{\ehat_{ji}(\gammai^d)} \geq  g_i(\gammai^*) - \alpha \sum_{j \ne i}{\ehat_{ji}(\gammai^*)}$, which implies:
    \begin{equation}\label{eq:dom_inequality}
         \sum_{j \ne i}{(\ehat_{ji}(\gammai^d) - \ehat_{ji}(\gammai^*))} \le \min\left(1, \tfrac{g_i(\gammai^d) - g_i(\gammai^*)}{\alpha}\right)
    \end{equation}
    
    Next, the social welfare expression (sum of all utilities) is unchanged due to transaction cost being revenue neutral; thus the social welfare difference between $S^*$ and $S^d$ is given by: $\sum_{i=1}^{n}g_i(\gammai^*) - \sum_{j \ne i} e_{ij}(\gammaj^*) - g_i(\gammai^d) + \sum_{j \ne i} e_{ij}(\gammaj^d)$ which is equal to:
    \begin{equation}\label{eq:offline_regret_decomp}
        \sum_{i=1}^{n}{\sum_{j \ne i}\left(\ehat_{ji}(\gammai^d) - \ehat_{ji}(\gammai^*)\right)} - \left(g_i(\gammai^d) - g_i(\gammai^*)\right) + \sum_{j \ne i}{b_{ji}(\gammai^d) - b_{ji}(\gammai^*)}
    \end{equation}
    which holds since we are summing all possible externality pairs for seller sets $\gamma^d$ and $\gamma^*$. Next, we apply inequality \ref{eq:dom_inequality} to upper bound the expression above:
    \begin{gather*}
        \text{(\ref{eq:offline_regret_decomp})} \leq \bmax + \sum_{i=1}^{n}{\min\left(1, \tfrac{g_i(\gammai^d) - g_i(\gammai^*)}{\alpha}\right) - (g_i(\gammai^d) - g_i(\gammai^*))}\\
        = \bmax + \sum_{i=1}^{n}{\min\left(1 - (g_i(\gammai^d) - g_i(\gammai^*)), 
        (g_i(\gammai^d) - g_i(\gammai^*))\left(\tfrac{1}{\alpha} - 1\right)\right)}
    \end{gather*}
    Putting aside the $\bmax$ term, each summand above is a min of two values, which is maximum when the two values are equal. Let $v_i = (g_i(\gammai^d) - g_i(\gammai^*))$. Thus we have $\forall i, 1 - v_i = \tfrac{v_i}{\alpha} - v_i \implies v_i = \alpha$, and thus we can upper bound this with: $\sum_{i=1}^{n}{1 - \alpha} = n(1 - \alpha)$. 
    
    \emph{Lower bound:} Consider a setting with only two sellers $\ell_1$ and $\ell_2$. For each buyer $i$, we have the following: $\forall i, \, g_i(\emptyset) = 0; g_i(\ell_1) = 0; g_i(\ell_2) = 1 - \alpha - \epsilon; g_i(\ell_1 \cap \ell_2) = 1$ and $\forall i,j, \, i \ne j, \, e_{ij}(\emptyset) = 0; e_{ij}(\ell_1) = 0; e_{ij}(\ell_2) = 0; e_{ij}(\ell_1 \cap \ell_2) = \tfrac{1}{n-1}$. Regarding the social welfare optimal solution, note that there is no reason to choose $\emptyset$ or $\ell_1$; suppose $S$ is such that $k$ buyers choose $\ell_1 \cap \ell_2$ and $n-k$ choose $\ell_2$. We thus have: $sw(S) = (n-k)\left[1 - \alpha - \epsilon - \tfrac{k}{n-1}\right] + k\left[1 - \tfrac{k}{n-1}\right]$ which is equal to $n - n\alpha - n\epsilon - k\left[\tfrac{n}{n-1} - \alpha - \epsilon \right]$. As $\epsilon \rightarrow 0$, social welfare is maximized when $k = 0$, since $\alpha \in [0,1]$ and $\tfrac{n}{n-1} > 1$. Thus the optimal is all agents choosing $\ell_2$, for a social welfare of $n(1-\alpha)$ as $\varepsilon \rightarrow 0$. As for equilibrium/dominant strategy, observe that each buyer's dominant strategy is choosing $\ell_1 \cap \ell_2$, \emph{regardless} of platform's bias. Choosing this option guarantees them at least $1-\alpha$ utility whereas choosing $\gamma = \ell_2$ gives them at most $1- \alpha - \varepsilon$ utility. However, this dominant strategy, however, has social welfare 0, leading to worst-case WRaE $\rightarrow n(1-\alpha)$ on this instance. 
\end{proof}

\section{Online setting with learned valuation}\label{section:learned_valuations}
We now relax the assumption that buyers know their expected valuations a priori; instead, we consider them repeatedly interacting with the market and acting upon values learned during this process. This is consistent with real-world behavior since (1) sellers usually offer their data products for a prescribed time-period with buyers needing to pay anew for continued access (typical in AWS Data Marketplace) and (2) data products often capture real-time phenomenon (weather or foot traffic data) leading consumers to routinely purchase fresh data from the same source (see Appendix \ref{app:online_justification} for a detailed discussion). Buyers in this online setting face an exploration vs. exploitation problem, a dichotomy well modeled by the multi-armed bandit (MAB) framework. Each time a buyer interacts with the market, they can choose between $2^k$ ``arms'', representing the different seller combinations at their disposal. By ``pulling an arm'' they choose one of these options and observe the stochastic gain, externality, and transaction cost associated with that choice. At each round $t$, all buyers make such a decision and note that buying nothing, $\gamma = \emptyset$, is a valid strategy for any buyer at any time. The additional challenges here are twofold: proposing a buyer learning algorithm that works well over time given the exponential number of arms, and showing this buyer-optimal algorithm does not degrade the social welfare regret. Our goal in this section is to explore the possible learning algorithms available to buyers and the scenarios wherein they are performant.

We consider the platform charging the proposed transaction cost $\mathcal{T}_i(S)$ at each round as described in definition \ref{definition:net_transaction_fee}. Under the proposed transaction cost, it is the dominant strategy of each buyer to maximize the expected effective utility, $\udi(\gamma) = \E[\Udi(\gamma)]$, with $\gammai^d$ being the maximizer, and $\Udi(\gamma)$ referring to its stochastic analogue. Together, $\Udi(\gamma)$ and $\udi(\gamma)$ are analogous to the stochastic and expected ``reward'' for an arm in the standard MAB setting. Using this, we now define appropriate notions of individual and collective regret for the online setting.



\begin{definition}\label{definition:dom_strat_regret}
    For instance $\I$, the expected \textbf{online effective regret} for buyer $i$ is the difference in the effective utility between their full information dominant strategy, $\gammai^d$, and their strategy at time $t$, $\gammai^t$:
    $R_d^i(T; \I) = \E\left[\sum_{t=1}^{T}\Udi(\gammai^d) - \Udi(\gammai^t)\right]$, which is expanded as: $R_d^i(T; \I) =$
    \begin{equation*}
        \sum_{t=1}^{T}{\left[g_i(\gammai^d) - g_i(\gammai^t) - \ab \sum_{j \ne i}{\ehat_{ji}(\gammai^d) - \ehat_{ji}(\gammai^t)}\right]}
    \end{equation*}
    \vspace{-1em}
\end{definition}

\begin{definition}\label{definition:social_welfare_regret}
    For an instance $\I$, the expected \textbf{online cumulative welfare regret} is the difference in total expected utility (across all agents) between the social welfare maximal strategy, $S^* = (\gamma_1^*, \dots, \gamma_n^*)$, and the strategy taken at time $t$, $S^t = (\gamma_1^t, \dots, \gamma_n^t)$. Since the proposed transaction cost is revenue-neutral, this regret be expressed as: $R_w(T; \I) =$
    \begin{equation*}
        \sum_{i=1}^{N}{
            \sum_{t=1}^{T} {
                \left[
                    g_i({\gammai^*}) - g_i(\gammai^t) - \sum_{j \ne i}{e_{ij}(\gammaj^*) - e_{ij}(\gammaj^t)}
                \right]
            }
        }
    \vspace{-1em}
    \end{equation*}
    
\end{definition}

\subsection{Algorithms}
The exponential number of seller sets available to each buyer make learning a challenging problem. Using the forth-coming lemma \ref{lemma:concentration}, one can invoke the Upper Confidence Bound (UCB) algorithm to obtain $\Tilde{\Oh}(\sqrt{2^k T})$ effective regret for each buyer. To improve upon this, additional structure is needed. At first glance, the linear/combinatorial bandit framework, wherein an agent can simultaneously pull a subset $k$ arms at each round, may seem attractive \citep{combes2015combinatorial, cesa2012combinatorial}. While this captures buyer interactions in our setting and reduces regret dependency from $2^k$ to $k$, it assumes a linear reward model. That is, an unknown vector $\bm{x} \in \reals^k$ specifies the individual reward for each of the $k$ options, and the reward for choosing a subset $\gamma$ is: $\gamma^T\bm{x}$. In our setting, each of the $k$ coordinates represents a different data seller, this becomes a strong assumption as it implies the net gain and externality experienced by buyers when purchasing from a bundle of sellers is uncorrelated and can be linearly decomposed. The utility of data products is argued to be richly correlated and this assumption may not hold \citep{agarwal2019marketplace}.

With the linear utility assumption being too strong, we consider a weaker structure. $\gamma$ can be represented as a $k$ bit string, with bit $i$ denoting inclusion of seller $\ell_i$. Then for any $\gamma^1$ and $\gamma^2$, define $D_h(\gamma^1, \gamma^2)$ as the normalized Hamming distance (range is $[0, \dots, 1]$, in increments of $1/k$) between the two, counting the number of sellers in which these sets differ. If $D_h(\gamma^1, \gamma^2)$ is close in Hamming distance, then the two seller sets consist of roughly the same sellers, and it is reasonable that the gain and externalities induced by these two will not be drastically different. We formalize this with the following \emph{metric} property: for each buyer $i$ and any pair $(\gamma^1, \gamma^2)$: $|g_i(\gamma^1) - g_i(\gamma^2)| \leq \lambda_g D_h(\gamma^1, \gamma^2)$ and $|\sum_{j}e_{ji}(\gamma^1) - \sum_{i}e_{ji}(\gamma^2)| \leq \lambda_e D_h(\gamma^1, \gamma^2)$\footnote{WLOG, we fix $\lambda_g = \lambda_e = 1$ as it only changes bounds in a multiplicative manner}. We denote $\Delta_i(\gammai) = \udi(\gammai^d) - \udi(\gammai)$ as the buyer $i$'s expected effective utility gap for $\gammai$, with $\Delta_i(\gammai) \in [0,2]$. $n_t(\gammai)$ denotes how often $\gammai$ has been selected by buyer $i$. The metric property implies that $\Delta_i(\gammai^t) = \udi(\gammai^*) - \udi(\gamma^1) \leq 2D_h(\gamma^1, \gamma^2)$. Lastly, no assumptions are made \emph{across} different buyers - i.e. no relation assumed between $g_i(\gamma^1)$ and $g_j(\gamma^2)$. 

The metric assumption is far weaker than linearity as it still allows rich levels of correlation between options. A common technique in metric bandits is to uniformly cover the space of arms and UCB over the arms in the cover \citep{slivkins2019introduction}. This is practical when arms are in a continuous metric space since covering can be arbitrarily dense; in Hamming space however, any center in the covering is at least $\tfrac{1}{k}$ away from an arm it covers, leading to degenerate bounds. \citet{kleinberg2019bandits} propose a general purpose \emph{zooming} algorithm, which adaptively discretizes a region proportional to its reward. It maintains a set of \emph{active} arms, each of which \emph{covers} any arm falling within its confidence radius. It ensures all arms are covered at every round, and selects from active arms using the UCB rule. Given the near-optimality of this algorithm in metric settings, we adapt this for our purpose \citep{kleinberg2019bandits}. While our proposed algorithm (given below) is similar to the original, obtaining good regret bounds requires more careful analysis. The value of $\gamma$ depends on both the gain and externality which requires a stronger concentration result. More importantly, distance in Hamming space is quantized. This means that any two distinct elements are at least $\frac{1}{k}$ apart and there are a large number of elements that cannot be strictly compared. While the zooming algorithm in pathologically worst instances cannot improve upon standard algorithms, we show that in many natural instances, it performs significantly better. Our insights may be of independent interest in bandit settings with exponential arms. 
We present the algorithm below and continue the analysis thereafter.

\begin{algorithm}
\caption{Zooming algorithm for buyer $i$}\label{alg:cap}
\begin{algorithmic}
\State Active set $\mathcal{A}_i$ randomly initialized with a single seller set
\State Confidence radius for each $\gamma \in \mathcal{A}_i$ is $c_t^i(\gamma) = \sqrt{\tfrac{12\log T}{n_t^i(\gamma) + 1}}$
\State{\textbf{for} $t=1,\dots,T$ \textbf{do}}
    \Indent 
        \State \texttt{// An active seller-set $\gamma' \in \mathcal{A}_i$ covers a set $\gamma$ if $\lambda D_h(\gamma, \gamma') \leq c_t^i(\gamma')$}
        \State{\textbf{if} there are $\gamma$ not covered (under Hamming distance) by seller-sets in $\mathcal{A}_i$}
        \Indent
            \State{Pick any uncovered choice and append it to $\mathcal{A}_i$}
        \EndIndent
        \State \textbf{End}
        \State{Define $\text{UCB}_i(\gamma) = \udihat(\gamma) + 2c^i_t(\gamma)$ for each active choice $\gamma$} 
        \State{Select the arm with the highest $\text{UCB}_i$ value.}
    \EndIndent
\State \textbf{End}
\end{algorithmic}
\end{algorithm}

\subsection{Online Effective Regret}
We bound the online effective regret of a buyer $i$ under the zooming algorithm. First, using McDiarmid's inequality, Lemma \ref{lemma:concentration} shows sampled effective utility $\udihat$ is concentrated around its expected quantity $\udi$ with high probability. Conditioned upon this \emph{clean} event, Lemma \ref{lemma:slivkins_lemma} bounds the number of times an active choice is selected and the distance between any two active choices, mirroring a result in \citet{kleinberg2019bandits}. Proofs for these two lemmas are in the Appendix \ref{app:learn}.

\begin{lemma}\label{lemma:concentration}
    Let $\udihat(\gamma, h)$ denote the sample mean of $\Udi(\gamma)$ with $h$ samples, and define event $\mathcal{E} = \{\forall i, \forall \gamma, \forall t,$ $ |\udihat(\gamma, n_i^t(\gamma)) - \udi(\gamma)| \leq c_t^i(\gamma) \}$. Then $Pr[\mathcal{E}] \geq 1 - 1/T^2$.
\end{lemma}

\begin{lemma}\label{lemma:slivkins_lemma}
    If $\mathcal{E}$ holds, then $n^i_t(\gamma) \leq \tfrac{108\log T}{\Delta_i^2(\gamma)}$, $\forall \gamma, i, t$. For two active sets $\gamma^1, \gamma^2 \in \mathcal{A}_i$, $D_h(\gamma^1, \gamma^2) \geq \max(1/k, 1/3 \min(\Delta_i(\gamma^1), \Delta_i(\gamma^2)))$.
\end{lemma}

With the following results, theorem \ref{theorem:main_regret} states a precise instance dependent regret bound for each buyer. To sketch this, all seller sets are partitioned based on how close their effective utility is to the optimal $\gammai^d$, and add up their contribution to the regret. For seller sets where the difference is large, we appeal to lemma \ref{lemma:slivkins_lemma} and use a packing argument to bound the number of active elements here. For any active $\gamma$, their contribution to the total regret for buyer $i$ can be expressed using their gap $\Delta_i(\cdot)$ and the number of times they are chosen, $n_t^i(\cdot)$. For seller sets with utility close to the optimal, the discrete property of the Hamming metric makes lemma \ref{lemma:slivkins_lemma} too coarse, and the exact bound is instance-dependent. We address this in subsequent results.

\begin{theorem}\label{theorem:main_regret}
     For instance $\I$ and buyer $i$, let $S^{\I}_i\left(r/k\right)$ be the set of active sellers sets where each element $\gamma \in S^{\I}_i\left(r/k\right)$ satisfies $r/k \leq \Delta_i(\gamma) \leq 2r/k$. Let $\mathcal{C}(\cdot, d)$ denote the maximum packing of a region with balls of Hamming diameter $d$. Then for some $\delta > 0$, expected dominant strategy regret $R^i_d(T, \I)$ is upper bounded by:
    \begin{equation*}
         \frac{\delta}{k}2T + 108k\log T \left(\sum_{j=-1}^{\log 1/\delta}{2^j \left|S^{\I}_{i}\left(\tfrac{1}{2^j k}\right)\right|} + \sum_{j=2}^{\log k}{2^{-j} |\mathcal{C}(S_i^{\I}(\tfrac{2^j}{k}), \tfrac{2^j}{3k})|} \right)
    \end{equation*}
\end{theorem}
\begin{proof}
     We consider the contribution of seller sets in $S^{\I}_i(\tfrac{r}{k})$ toward the total effective regret suffered by buyer $i$. By adding this contribution over all values of $\tfrac{r}{k}$ (with $r < k$), we can express the cumulative regret. For each active seller set $\gamma \in S_i^{\I}(\tfrac{r}{k})$, their contribution toward DS regret is $\Delta_i(\gamma) \cdot n^i_t(\gamma)$. Next, assume the clean event $\mathcal{E}$ holds. By lemma \ref{lemma:slivkins_lemma}, we have that $\Delta_i(\gamma) \leq 3 c_t(\gamma) = 3 \sqrt{\tfrac{12\log T}{n^i_t(\gamma) + 1}}$ which implies $n_t(\gamma) \leq \tfrac{108\log T}{\Delta_i^2(\gamma)}$. Thus, $\Delta_i(\gamma) \cdot n_t^i(\gamma) \leq \tfrac{108\log T}{\Delta_i(\gamma)} \leq \tfrac{108 k \log T}{r}$. Now consider three ranges of $\tfrac{r}{k}$: (1) $\tfrac{r}{k} < \tfrac{\delta}{k}$, for some $\delta \in (0,1)$, (2) $\tfrac{r}{k} \in \left[\tfrac{\delta}{k}, \tfrac{2}{k}\right]$, and (3) $\tfrac{r}{k} \in \left[\tfrac{4}{k}, 1\right]$. For (1), we will use a trivial upper bound and will leave $\delta$ to be appropriately selected later. For (2), observe that by lemma \ref{lemma:slivkins_lemma}, any two active seller sets satisfy, $D_h(\gamma^1, \gamma^2) \geq \tfrac{1}{3}\min(\Delta(\gamma^1), \Delta(\gamma^2)) = \tfrac{c}{3k}$, where $c \leq 2$. In other words, the lower bound for the distance between any active seller is $< \tfrac{1}{k}$ and thus in the worst case, each seller sets only covers itself. For the elements in (3) however, applying lemma \ref{lemma:slivkins_lemma} implies that any two active sets are at least $\tfrac{r}{3k}$ apart. The maximum number of active seller sets here is thus upper bounded by the maximal packing of the space $S_i^{\I}(\tfrac{r}{k})$ with balls of diameter $\tfrac{r}{3k}$. Thus, the cumulative regret can be expressed as follows ($r$ is expressed as powers of two: $r = 2^j$).
     \vspace{-1.2ex}
     \begin{equation}
        R^i_d(T, \I) \leq \frac{\delta}{k}2T + \sum_{j=-1}^{\log \tfrac{1}{\delta}}{108 k \log T k 2^j \left|S^{\I}_{i}\left(\tfrac{1}{2^j k}\right)\right|} + \sum_{j=2}^{\log k}{108 k2^{-j} \log T  |\mathcal{C}(S_i^{\I}(\tfrac{2^{j}}{k}), \tfrac{2^{j}}{3k})|}
    \end{equation}
    With this, we now match the bound in the theorem statement. Denote this by $R(\cdot|\text{clean})$, which captures the conditioning event. The bad event is lemma \ref{lemma:concentration} not holding (probability $\tfrac{2}{T^2}$). Thus, the expected regret is dominated by $R(\cdot|\text{clean})$ since: $R_d^i(T, \I) \leq R(\cdot|\text{clean})(1 - \tfrac{2}{T^2}) + 2T \tfrac{2}{T^2}$.
\end{proof}

Despite being stated from an instance-dependent perspective, this result is quite general and insightful. Observe that the worst case occurs when all seller sets have gap $\Delta_i(\gamma) < 1/k$; the Hamming metric is not useful here since its smallest non-zero value is $1/k$. In this pathological case, the regret and the algorithm devolve into UCB (Corollary \ref{cor:worst_case_ucb}). The algorithm shines, however, when the gap to the optimal is spread out evenly as a function of the Hamming distance (Theorem \ref{theorem:good_metric_bound}).

\begin{cor}\label{cor:worst_case_ucb}
    For any metric instance $\I$ where $\Delta_i(\gamma) \leq 1/k$ for all $\gamma$, $R^i_d(T; \I) \leq  \Oh(\sqrt{2^k T\log T})$.
\end{cor}
\begin{proof}
    Since the gap $\Delta_i(\gamma) \leq \tfrac{1}{k}$, the third term in the expression given in theorem \ref{theorem:main_regret} will be 0, and we can focus on the first two terms. Next, we note that the second term, $\sum_{j=-1}^{\log \tfrac{1}{\delta}}{k 2^j108\log T\left|S^{\I}_{i}\left(\tfrac{1}{2^j k}\right)\right|}$, is increasing in $j$ and maximized when all elements belong to $S_i^I(\delta/k)$. This can be expressed as: $S_i^I(\delta/k) = \sum_{j=-1}^{\log \tfrac{1}{\delta}}{k 2^j108\log T\left|S^{\I}_{i}\left(\tfrac{1}{2^j k}\right)\right|} \leq k 2^k \tfrac{1}{\delta}108 \log T$. Going back to cumulative regret, we now have two terms increasing and decreasing in $\delta$ respectively. Their sum is minimized when we set them to be equal:
    \begin{equation*}
        \begin{split}
            \tfrac{\delta}{k}2T = k 2^k \tfrac{1}{\delta}108 \log T \implies 2\delta = \sqrt{\tfrac{k^2 2^k 108\log T}{T}} \implies R^i_d(T) = \Oh\left(\sqrt{2^k T \log T}\right) 
        \end{split}
    \end{equation*}
\end{proof}

\begin{theorem}\label{theorem:good_metric_bound}
    For any metric instance $\I$ where $\Delta_i(\gamma) \in [2D_h(\gammai^d, \gamma) - 2/k, 2D_h(\gammai^d, \gamma)]$ \footnote{2 is just a scale factor since $D_h(\cdot) \leq 1$ but $\Delta_i(\gamma) \in [0,2]$} for all $\gamma$, $R^i_d(T; \I) \leq  \Oh\left(k(\sqrt{kT\log T} + 2^{0.58k}\log T)\right)$
\end{theorem}
\begin{proof}
    Starting with the expression given in theorem \ref{theorem:main_regret}, we initially focus on the first two terms. Note that the second term considers seller sets whose gap is less than $\tfrac{4}{k}$. By given property of $\I$, seller sets who satisfy this are at most $\tfrac{3}{k}$ Hamming distance away. Thus, we can write the second term as: $\sum_{j=-1}^{\log 1/\delta}{k 2^j108\log T\left|S^{\I}_{i}\left(\tfrac{1}{2^j k}\right)\right|}$ which we can upper bound as: $\leq \sum_{j=-1}^{\log 1/\delta}{k 2^j108\log T \sum_{i = 0}^{3}{\binom{k}{i}}} \leq 216\log(T) k^4 \frac{1}{\delta}$. Comparing this with the $\tfrac{\delta}{k}2T$ term, we now have an increasing and decreasing term in $\delta$, which is minimized when we set them to be equal: $\tfrac{\delta}{k}2T = k^4 \tfrac{1}{\delta}216 \log T \implies \delta = \sqrt{\tfrac{k^5 216 \log T}{2T}}$. Note that $\delta$ is not an algorithm or instance parameter but used only for analysis. Putting all this together and using a standard assumption that $T > k$, we can state the first two term to be upper bounded by $\Oh(k^{3/2}\sqrt{T \log T})$. 
    
    We now turn our attention to the third and last term which considers regions $S_i^\I(\tfrac{r}{k})$ with $\tfrac{r}{k} \geq \tfrac{4}{k}$. Our instance condition implies the set of choices whose gap is in region $S_i^\I(\tfrac{r}{k})$ correspond to $\gamma$ satisfying $D_h(\gammai^d, \gamma) \in [\tfrac{r}{2k}, \tfrac{r}{k} + \tfrac{1}{k}]$. We are looking to bound the maximal packing size of this region with balls of Hamming diameter $\tfrac{r}{3k}$. By a simple volume argument, we have that: $|\mathcal{C}(S_i^{\I}(\tfrac{r}{k}), \tfrac{r}{3k})| \leq \frac{\sum_{j=r/2}^{r}{\binom{k}{j}}}{\sum_{j=0}^{r/6}{\binom{k}{j}}}$. Next, note that we wish to upper bound the sum of the packing sizes of for each $S_i^{\I}(\tfrac{r}{k})$ region with $\tfrac{r}{k} \in [\tfrac{4}{k}, \tfrac{8}{k}, \dots, \tfrac{1}{2}, 1]$. These correspond to seller sets being in the following Hamming distance ranges from $\gammai^d$:
    $[[\tfrac{2}{k}, \tfrac{5}{k}], [\tfrac{4}{k}, \tfrac{9}{k}] \dots, [\tfrac{1}{4}, \tfrac{1}{2} + \tfrac{1}{k}], [\tfrac{1}{2}, 1]]$. We observe that the packing size upper bound (given above) is increasing in $r$ up to $r=\tfrac{1}{2}$, which we express as:
    \begin{equation}
        |\mathcal{C}(S_i^{\I}(\tfrac{1}{2}), \tfrac{1}{6})| \leq \frac{\sum_{j=k/4}^{k/2}{\binom{k}{j}}}{\sum_{j=0}^{k/12}{\binom{k}{j}}} \leq \frac{2^{k-1}}{2^{H\left(\tfrac{1}{12}\right)k}} \leq 2^{0.58k}
    \end{equation}
    where $H(\cdot)$ is the binary entropy function. Now we can sum over all the values of $\tfrac{r}{k}$ under consideration and express this as follows: $108 k \log T \sum_{j=2}^{\log k}2^{-j} |\mathcal{C}(S_i^{\I}(\tfrac{2^{j}}{k}), \tfrac{2^{j}}{3k})| \leq 108 k 2^{0.58k} \log T$. Putting this together with the bound for the first two terms gives us the desired result.
\end{proof}

The result above improves upon UCB, disentangling the dominant $\sqrt{T}$, from the exponential $2^{0.58k}$ term. The instance condition that led to this can be naturally interpreted: as seller sets $\gamma$ become more distinct from the buyer's optimal $\gammai^d$, a quantity captured by the Hamming metric, the gain and externality of this diminishes accordingly. In other words, options close to the optimal have similar gain and externality, with those further away getting poorer. Note that $\gammai^d$ need not be known, and this only assumes a relationship with respect to $\gammai^d$ and not between any arbitrary arms. This condition can be viewed as a looser form of the ``target-set'' assumption in bandit literature \citep{slivkins2019introduction, kleinberg2019bandits}. In general, our goal here was to sketch out possible learning algorithms for buyers to learn their strategy, with the exponential set of options available being a challenge. While linear bandits provide the best regret guarantees, they require exact independence and linearity assumptions which may be unreasonable for data products. Correspondingly, we considered a weaker metric assumption and adapted a state-of-the-art algorithm for our unique Hamming metric setting. Our analysis shows that while in the worst case it cannot outperform UCB algorithm due to the discrete property of the Hamming metric, it can in more natural settings. Specifically, as gain and externality gap between an arbitrary seller set and buyer $i$'s optimal becomes more correlated with their Hamming distance, covering becomes more effective, and regret drops. In general, the zooming algorithm allows us to smoothly interpolate the regret as the structural assumptions of the instance become looser or tighter. UCB and Linear bandits can be seen as extreme versions of these, with no assumptions and very strong assumptions respectively. Lastly, from an operational perspective, buyers rarely consider all $k$ sellers in a marketplace; instead, they are usually deciding between a smaller subset of sellers. In such a case, a buyer only needs to learn valuations over this smaller set, and the regret bound becomes proportional to that instead of $k$.

\subsection{Online welfare regret}
We observed in theorem \ref{theorem:wrae_offline} that agents playing their dominant strategy under the proposed transaction cost leads to $\Theta(n(1-\alpha))$ welfare regret (WRaE). Having established that buyers can adopt online learning algorithms to learn their dominant strategy, we consider the analogous question in the online setting: what is the WRaE when agents play their dominant strategy using online algorithms? We prove that this can be decomposed into online effective regret discussed above and the offline WRaE (Theorem \ref{theorem:wrae_offline}): 

\begin{theorem}\label{theorem:online_welfare_regret}
    The online welfare regret $R_w(T;\I)$ is upper-bounded by: $2nT(1-\alpha) + 2T\bmax + n R_d(T;\I)$.
\end{theorem}
\begin{proof}
    Let $S^* = (\gamma_1^*, \dots, \gamma_n^*)$ be the social optimal strategy, $S^d = (\gamma_1^d, \dots, \gamma_n^d)$ the dominant strategy of each buyer, and $S^t = (\gamma_1^t, \dots, \gamma_n^t)$ the strategy taken by buyers at time $t$. We can add and subtract the welfare at the dominant strategy to the online welfare regret, $R_w(T)$, and express it as:
    \begin{equation*}
        \begin{split}
            R_w(T) = \sum_{t=1}^{T}{
            \sum_{i=1}^{N} {
                \left(
                    g_i({\gammai^*}) - \textstyle\sum{e_{ij}(\gammaj^*)} - g_i(\gammai^d) + \textstyle\sum{e_{ij}(\gammaj^d)}
                \right)
                + 
                \left(
                    g_i({\gammai^d}) - \textstyle\sum{e_{ij}(\gammaj^d)} - g_i(\gammai^t) + \textstyle\sum{e_{ij}(\gammaj^t)}
                \right)
            }
        } 
        \end{split}
    \end{equation*}
    The first part of the sum, $\sum_{i=1}^{N}{g_i({\gammai^*}) - \textstyle\sum{e_{ij}(\gammaj^*)} - g_i(\gammai^d) + \textstyle\sum{e_{ij}(\gammaj^d)}}$, is exactly the WRaE quantity that we upper bounded to $n(1-\alpha) + \bmax$ in theorem \ref{theorem:wrae_offline}, which we can then sum over $T$. Next, observe that $e_{ji}(\gamma) = \alpha \ehat_{ji}(\gamma) + (1-\alpha) \ehat_{ji}(\gamma) + b_{ji}(\gamma)$ for any $\gamma$. Thus, we can re-arrange the latter sum as:
    \begin{gather}
        \sum_{i=1}^{N}\left[\sum_{t=1}^{T}g_i(\gammai^d) - \sum_{t=1}^{T}\sum_{j\ne i}{\alpha \ehat_{ji}(\gammai^d)} - \sum_{t=1}^{T}{g_i(\gammai^t) + \sum_{t=1}^{T}\sum_{j \ne i}{\alpha \ehat_{ji}(\gammai^t)}}\right] \\ + T\sum_{j \ne i}{{[b_{ji}(\gammai^t) - b_{ji}(\gammai^d)}]} + T(1-\alpha)\sum_{j \ne i}{{[e_{ji}(\gammai^t) - e_{ji}(\gammai^d)}]}
    \end{gather}
    We note that the first half of the summand is equal to online effective regret for all buyers, and is equivalent to $nR_d(T;\I)$. The bias terms are upper bounded by $T\bmax$ and since total externality induced by any buyer is at most 1, the last summation term is upper bounded by $nT(1-\alpha)$. Thus, the total online welfare regret is upper-bounded by $2nT(1-\alpha) + 2T\bmax + nR_d(T;I)$. 
\end{proof}

We make a few remarks before concluding the online section. First, any horizon ($T$) dependent online algorithm for buyers can be converted to a horizon-independent version using the simple \emph{doubling trick} \cite{cesa2006prediction}. This entails running the algorithm in multiple phases, and in each phase $i$, an instance of the algorithm is executed for $T=2^i$ rounds. This has another beneficial property if we consider adjusting the intervention parameter $\alpha$ over the phases as well. Observe that re-arranging the theorem above (ignoring the bias terms) to equate $nT(1-\alpha) = nR_d(T;\I)$, yields $\alpha = \Oh\left(1 - \tfrac{R_d(T)}{T}\right)$. For any sub-linear learning algorithm, the ratio $\tfrac{R_d(T)}{T}$ is initially high but diminishes as the phases increase. In other words, for early phases when there is significant regret due to lack of learning, $\alpha$ can be small. In later phases, increasing $\alpha$ minimizes the WRaE that the dominant strategy causes. As such, the platform does not intervene strongly in early rounds when buyers know little and are simply exploring and only does so when buyers have learned. The online welfare regret under this schedule is asymptotically equivalent to the total effective regret faced by buyers. Overall, this is a nice operational property of our intervention and holds for any learning algorithm.

\subsection{Online welfare regret}
We observed in theorem \ref{theorem:wrae_offline} that agents playing their dominant strategy under the proposed transaction cost leads to $\Theta(n(1-\alpha))$ welfare regret (WRaE) upto prediction errors. Having established that buyers can adopt online learning algorithms to learn and play their dominant strategy, we now consider the analogous question in the online setting: what is the WRaE when agents play their dominant strategy under our proposed intervention using online algorithms? We prove this can be decomposed into the online effective regret discussed above and the offline WRaE proved in theorem \ref{theorem:wrae_offline}. 

\begin{theorem}\label{theorem:online_welfare_regret}
    For an instance $\I$, the online welfare regret, $R_w(T)$, is upper-bounded by: $2nT(1-\alpha) + 2T\bmax + n R_d(T;\I)$.
\end{theorem}
\begin{proof}
    Let $S^* = (\gamma_1^*, \dots, \gamma_n^*)$ be the social optimal strategy, $S^d = (\gamma_1^d, \dots, \gamma_n^d)$ the dominant strategy of each buyer, and $S^t = (\gamma_1^t, \dots, \gamma_n^t)$ the strategy taken by buyers at time $t$. We can add and subtract the welfare at the dominant strategy to the online welfare regret, $R_w(T)$, and express it as:
    \begin{equation*}
        \begin{split}
            R_w(T) = \sum_{t=1}^{T}{
            \sum_{i=1}^{N} {
                \left(
                    g_i({\gammai^*}) - \textstyle\sum{e_{ij}(\gammaj^*)} - g_i(\gammai^d) + \textstyle\sum{e_{ij}(\gammaj^d)}
                \right)
                + 
                \left(
                    g_i({\gammai^d}) - \textstyle\sum{e_{ij}(\gammaj^d)} - g_i(\gammai^t) + \textstyle\sum{e_{ij}(\gammaj^t)}
                \right)
            }
        } 
        \end{split}
    \end{equation*}
    The first part of the sum, $\sum_{i=1}^{N}{g_i({\gammai^*}) - \textstyle\sum{e_{ij}(\gammaj^*)} - g_i(\gammai^d) + \textstyle\sum{e_{ij}(\gammaj^d)}}$, is exactly the WRaE quantity that we upper bounded to $n(1-\alpha) + \bmax$ in theorem \ref{theorem:wrae_offline}, which we can then sum over $T$. Next, observe that $e_{ji}(\gamma) = \alpha \ehat_{ji}(\gamma) + (1-\alpha) \ehat_{ji}(\gamma) + b_{ji}(\gamma)$ for any $\gamma$. Thus, we can re-arrange the latter sum as:
    \begin{gather}
        \sum_{i=1}^{N}\left[\sum_{t=1}^{T}g_i(\gammai^d) - \sum_{t=1}^{T}\sum_{j\ne i}{\alpha \ehat_{ji}(\gammai^d)} - \sum_{t=1}^{T}{g_i(\gammai^t) + \sum_{t=1}^{T}\sum_{j \ne i}{\alpha \ehat_{ji}(\gammai^t)}}\right] \\ + T\sum_{j \ne i}{{[b_{ji}(\gammai^t) - b_{ji}(\gammai^d)}]} + T(1-\alpha)\sum_{j \ne i}{{[e_{ji}(\gammai^t) - e_{ji}(\gammai^d)}]}
    \end{gather}
    We note that the first half of the summand is equal to online effective regret for all buyers, and is equivalent to $nR_d(T;\I)$. The bias terms are upper bounded by $T\bmax$ and since total externality induced by any buyer is at most 1, the last summation term is upper bounded by $nT(1-\alpha)$. Thus, the total online welfare regret is upper-bounded by $2nT(1-\alpha) + 2T\bmax + nR_d(T;I)$. 
\end{proof}

We make a few remarks before concluding the online section. First, any horizon ($T$) dependent online algorithm for buyers can be converted to an horizon-independent version using the simple \emph{doubling trick} \cite{cesa2006prediction}. This entails running the algorithm in multiple phases, and in each phase $i$, an instance of the algorithm is executed for $T=2^i$ rounds. This has another beneficial property if we consider adjusting the intervention parameter $\alpha$ over the phases as well. Observe that re-arranging the theorem above (ignoring the bias terms) to equate $nT(1-\alpha) = nR_d(T;\I)$, yields $\alpha = \Oh\left(1 - \tfrac{R_d(T)}{T}\right)$. For any sub-linear learning algorithm, the ratio $\tfrac{R_d(T)}{T}$ is initially high but diminishes as the phases increase. In other words, for early phases when there is significant regret due to lack of learning, $\alpha$ can be small. In later phases, increasing $\alpha$ minimizes the WRaE that the dominant strategy causes. As such, the platform does not intervene strongly in early rounds when buyers know little and are simply exploring and only does so when buyers have learned. The online welfare regret under this schedule is asymptotically equivalent to the total effective regret faced by buyers. Overall, this is a nice operational property of our intervention and holds for any online learning algorithm that a buyer adopts.

\section{A Richer Externality Model}\label{section:rich_externality}
With a complete picture of both the learning and game dynamics, we now consider an extension by analyzing a richer externality model, termed \emph{joint externality}. Under this externality model, when buyer $i$ purchases $\gammai$ and buyer $j$ purchases $\gammaj$, the externality buyer $i$ suffers depends on both buyer $i$'s decision, $\gammai$, and $j$'s decision $\gammaj$: $E_{ij}(\gammai, \gammaj)$. To clarify notation, for two specific seller sets $\gamma^1$ and $\gamma^2$, $e_{ij}(\gamma^1, \gamma^2)$ implies buyer $i$ owning $\gamma^1$ and buyer $j$ owning $\gamma^2$. So in general, $e_{ij}(\gamma^1, \gamma^2) \ne e_{ij}(\gamma^2, \gamma^1)$, as the ownership is reversed. However, if we write $e_{ij}(\gammai, \gammaj)$ the arguments already imply which buyer owns which set, and thus notationally $e_{ij}(\gammai, \gammaj) = e_{ij}(\gammaj, \gammai)$. The total expected externality suffered by buyer $i$ from all other buyers is thus: $\sum_{j \ne i}{e_{ij}(\gammai, \gammaj)}$. Joint externality can also be represented in matrix form. For a pair $i,j$ with $i < j$, the matrix $\Es_{ij}$ tabulates the externality suffered by $i$ due to $j$ with $\Es_{ij}[\ell, k] = e_{ij}(\gammai=\gamma^\ell, \gammaj=\gamma^k)$; similarly, $\Es_{ji}[\ell, k] = e_{ji}(\gammai=\gamma^\ell, \gammaj=\gamma^k)$. 

Observe that this is a richer class of externality than previously studied and can capture more diverse settings (see appendix \ref{appendix:externality} for a detailed discussion). To our knowledge, this work is the first to explore this more general externality model in data markets. In this section, we focus on how this affects the game dynamics with and without our proposed intervention. This richer setting naturally leads to weaker guarantees, and as such we define a weaker notion of pure equilibrium. We then move on to our first pair of results that paint a bleak picture of the Nash equilibrium properties under the joint externality model without any intervention: even a good approximate equilibrium may not always exist, and even in settings where it does, the welfare maybe very poor \footnote{The worst-case instance in Proposition \ref{thm:joint_noeq_wrae} can't be captured by the independent model. Thus with no intervention, joint externality model does not aid the welfare properties of pure equilibrium.}.

\begin{definition}
    An action profile $S^q = (\gamma^q_1, \dots, \gamma^q_n)$ is an $\varepsilon$- pure Nash Equilibrium if no agent $i$ can increase their utility by more than $\varepsilon$ by unilaterally deviating from their strategy $\gamma^q_i$. When $\varepsilon = 0$, this coincides with a PNE.
\end{definition}

\begin{prop}\label{thm:joint_noeq}
    Without any intervention, there are instances wherein no $\varepsilon$-pure Nash Equilibrium exists for $\varepsilon < 1$.
\end{prop}
\begin{proof}
    Consider 2 buyers, and let $\gamma^a$ and $\gamma^b$ denote two arbitrary orders with $a \in \{1, \dots, 2^k\}$, $b \in \{1, \dots, 2^k\}$ and $a \ne b$. Then, let the externality involving $\gamma^a$ and $\gamma^b$ be as follows (note $\Es^{ab}_{12}$ is a sub-matrix of $\Es_{12}$):
    \begin{equation*}
        \Es^{ab}_{12} =
        \begin{bmatrix}
            e_{12}(\gamma_1=\gamma^a, \gamma_2=\gamma^a) = 1  &  e_{12}(\gamma_1=\gamma^a, \gamma_2=\gamma^b) = 0      \\
            e_{12}(\gamma_1=\gamma^b, \gamma_2=\gamma^a) = 0  &  e_{12}(\gamma_1=\gamma^b, \gamma_2=\gamma^b) = 1      
        \end{bmatrix}
    \end{equation*}
    \begin{equation*}
        \Es^{ab}_{21} =
        \begin{bmatrix}
            e_{21}(\gamma_1=\gamma^a, \gamma_2=\gamma^a) = 0  &  e_{21}(\gamma_1=\gamma^a, \gamma_2=\gamma^b) = 1     \\
            e_{21}(\gamma_1=\gamma^b, \gamma_2=\gamma^a) = 1  &  e_{21}(\gamma_1=\gamma^b, \gamma_2=\gamma^b) = 0      
        \end{bmatrix}
    \end{equation*}
    Let the expected externality $e_{12}$ and $e_{21}$ be 1 for any other seller set pair not included above --- that is, the full externality matrices for this example, $\Es_{12}$ and $\Es_{21}$, are rank 2. Next, let the expected net gain be as follows: $g_i(\gamma^a) = g_i(\gamma^b) = 1$ for $i \in \{1,2\}$, with the expected net gain for all other seller sets $\gamma \notin \{\gamma^a, \gamma^b\}$ being 0. This clearly implies that it is strictly dominant for both buyers to choose between $\gamma^a$ or $\gamma^b$ over any other choices as that leads to $-1$ utility. Thus, it suffices to consider the utility for both buyers involving only these states which we can express in a matrix as follows (buyer 1 is the row player, and buyer 2 is the column player):
    \begin{equation*}
        \begin{bmatrix}
            u_1(\gamma^a, \gamma^a), u_2(\gamma^a, \gamma^a) & u_1(\gamma^a, \gamma^b), u_2(\gamma^a, \gamma^b) \\
            u_1(\gamma^b, \gamma^a), u_2(\gamma^b, \gamma^a) & u_1(\gamma^b, \gamma^b), u_2(\gamma^b, \gamma^b) \\
        \end{bmatrix} = 
        \begin{bmatrix}
            0, 1 & 1, 0 \\
            1, 0 & 0, 1 \\
        \end{bmatrix}
    \end{equation*}
    It is evident that for any action profile $\gamma_1, \gamma_2$ involving $a$ and $b$, one of the buyers will benefit an amount 1 by deviating. Since any other strategy involving the buyers choosing something beyond options $a, b$ is strictly dominated also by utility 1, the statement holds.
\end{proof}

\begin{prop}\label{thm:joint_noeq_wrae}
    Without any intervention, if a joint externality instance does have pure Nash Equilibria, then the WRaE of any such equilibria can approach $\rightarrow n$.
\end{prop}
\begin{proof}
    Consider $n$ buyers, and let $\gamma_a$ and $\gamma_b$ denote two arbitrary orders with $a \in \{1, \dots, 2^k\}$, $b \in \{1, \dots, 2^k\}$ and $a \ne b$. Buyers 1 through $n-1$ share the same externality whereas buyer $n$ is different. Let the externality involving options $\gamma^a$ and $\gamma^b$ be as follows:
     \begin{equation*}
        \forall i \in \{1, n-1\}, \forall j \,\,: \Es^{ab}_{ij} =
        \begin{bmatrix}
            e_{ij}(\gamma_1=\gamma^a, \gamma_2=\gamma^a) = \tfrac{1-\varepsilon}{n-1}  &  e_{ij}(\gamma_1=\gamma^a, \gamma_2=\gamma^b) = \tfrac{\varepsilon}{n-1}      \\
            e_{ij}(\gamma_1=\gamma^b, \gamma_2=\gamma^a) = \tfrac{1-\varepsilon}{n-1}  &  e_{ij}(\gamma_1=\gamma^b, \gamma_2=\gamma^b) = \tfrac{2 \varepsilon}{n-1}      
        \end{bmatrix} 
    \end{equation*}
    \begin{equation*}
        \forall j \,\,: \Es^{ab}_{nj} =
        \begin{bmatrix}
            e_{ij}(\gamma_1=\gamma^a, \gamma_2=\gamma^a) = \tfrac{1-\varepsilon}{n-1}  &  e_{ij}(\gamma_1=\gamma^a, \gamma_2=\gamma^b) = \tfrac{1}{n-1}      \\
            e_{ij}(\gamma_1=\gamma^b, \gamma_2=\gamma^a) = \tfrac{1}{n-1}  &  e_{ij}(\gamma_1=\gamma^b, \gamma_2=\gamma^b) = 0      
        \end{bmatrix}
    \end{equation*}

    Let the expected externality for all $(i,j)$ be 1 for any other seller set pair not including either $\gamma^a$ or $\gamma^b$. Next, let the expected net gain be as follows: $g_i(\gamma^a) = g_i(\gamma^b) = 1$ for $i \in \{1,2\}$, with the expected net gain for all other seller sets $\gamma \notin \{\gamma^a, \gamma^b\}$ being 0. As such, it is strictly dominant for \emph{any} buyer to select $\gamma^a$ or $\gamma^b$ over any other since the former will lead to a utility of at least 0, whereas the latter has utility $-1$. Thus it suffices to consider strategies pertaining only to $\gamma^a$ and $\gamma^b$ for the purpose of equilibrium. Next, note that for agents $1$ through $n-1$, choosing $\gamma^a$ is in fact strictly dominant since they suffer less externality as compared to choosing $\gamma^b$, regardless of whether buyer $n$ choosing $\gamma^a$ or $\gamma^b$. Since the first $n-1$ agents will always choose $\gamma^a$ as their dominant strategy, the only pure equilibrium consists of buyer $n$ choosing $\gamma^a$ as this has strictly lower utility. Thus, all agents choosing $\gamma^a$ is the only PNE, and we note that this has social welfare $n \varepsilon$. Next, consider all buyers choosing $\gamma^b$. This leads to social welfare $1 + (1-2\varepsilon)(n-1) = n - 2\varepsilon(n-1)$. Thus as $\varepsilon \rightarrow 0$, the best case WRaE $\rightarrow n$. 
\end{proof}

These two results show that the lack of intervention is doubly bad in this richer externality model. Not only can pure equilibrium lead to maximal welfare regret, paralleling the result in theorem \ref{thrm:independent_externality_eq}, but even good approximate equilibrium may not exist in this setting. We now consider introducing our transaction cost based on predicted externalities, which in this setting is equivalent to $\mathcal{T}_i(S) = \alpha(\sum_{j \ne i}{\Ehat_{ji}(\gammai, \gammaj)} - \Ehat_{ij}(\gammai, \gammaj))$; as before $\Ehat_{ij}(\gammai, \gammaj)$ is the predicted externality, and $b_{ij}(\gammai, \gammaj)$ is the bias. We first show that under this cost, an $\varepsilon$ PNE is always possible, with $\varepsilon \rightarrow 0$ linearly as $\alpha \rightarrow 0.5$ and $\Oh(b) \rightarrow 0$, where $\Oh(b)$ denotes purely bias terms.

\begin{theorem}\label{theorem:joint_eps_pne}
    Under the proposed transaction cost, there always exists an $\varepsilon$ PNE with $\varepsilon = 2|\alpha - 0.5|\sum_{j \ne i}{\max_{\gammai, \gammaj}|\ehat_{ji}(\gamma_i, \gamma_j) - \ehat_{ij}(\gamma_i, \gamma_j)|} + \sum_{j \ne i}{\max_{\gammai, \gammaj}{|b_{ij}(\gammai, \gammaj) - b_{ji}(\gammai, \gammaj)|}}$.
\end{theorem}
\begin{proof}
    For a general instance $\I = (g_i, e_{ij}, b_{ij}, \alpha)$, first consider a simpler version of this instance, $\tilde{\I} = (\tilde{g}, \tilde{e}_{ij}, \tilde{b}_{ij}, \tilde{\alpha})$ where $\tilde{g}_i = g_i$, $\tilde{e}_{ij}(\gammai, \gammaj) = \tfrac{1}{2}\left(e_{ij}(\gammai, \gammaj) + e_{ji}(\gammai, \gammaj)\right)$, $\tilde{\alpha} = 0$, and $\tilde{b}_{ij}(\cdot, \cdot) = 0$ for all $i, j$ and possible pairs of $(\gammai, \gammaj)$. We first note that in instance $\tilde{\I}$, externality are symmetric - i.e. $\tilde{e}_{ij}(\gammai, \gammaj) = \tilde{e}_{ji}(\gammai, \gammaj) \, \forall i,j$. We now show that the dynamics of instance $\tilde{I}^s$ can be represented through a potential function $\Phi(S)$ defined as follows: $\Phi(S) = \sum_{i}{\tilde{g}_i(\gammai)} - 0.5\sum_{i}\sum_{j}{\tilde{e}_{ij}(\gammai, \gammaj)}$. Without loss of generality, for a joint action profile $S = (\gamma_1, \dots, \gamma_n)$, consider buyer $1$ changing her action to $\gamma'_1$ and leading to profile $S'$. The resulting change in expected utility for her is given by $\tilde{g}_1(\gamma_1) - \tilde{g}_1(\gamma'_1) - \left(\sum_{j \ne 1}{\tilde{e}_{1j}(\gammai,\gammaj) - \tilde{e}_{1j}(\gammai', \gammaj)}\right)$. Note that the potential function difference is given by $\Phi(S) - \Phi(S') = \tilde{g}_1(\gamma_1) - \tilde{g}_1(\gamma'_1) - 0.5\sum_{j \ne 1}{\tilde{e}_{1j}(\gamma_1, \gammaj)} - 0.5\sum_{j \ne 1}{\tilde{e}_{j1}(\gamma_1, \gammaj)} + 0.5\sum_{j \ne 1}{\tilde{e}_{1j}(\gamma'_1, \gammaj)} - 0.5\sum_{j \ne 1}{\tilde{e}_{j1}(\gamma'_1, \gammaj)}$. Since externalities are symmetric in $\tilde{I}$, this difference in potential functions is equivalent to the difference in agent 1's utility, making this a potential game, which implies game instance $\tilde{I}$ has a pure Nash Equilibrium.
    
    Let $S = (\gamma_1, \dots, \gamma_n)$ denote an action profile. Under $\tilde{\I}$, observe that utility for buyer $i$ can be expressed as $u_i(S; \tilde{\I}) = g_i(\gammai) - 0.5\sum_{j \ne i}{e_{ij}(\gammai, \gammaj)} - 0.5\sum_{j \ne i}{e_{ji}(\gammai, \gammaj)}$. Observe that for any strategy $S$, we can relate $u_i(S; \tilde{\I})$ to the utility for buyer $i$ in the general instance $\I$ under the same actions $S$ as follows: $u_i(S; \tilde{\I}) - u_i(S; \I)$
    \begin{equation}\label{equation:simple_real_equivalence}
        = (\alpha - 0.5) \sum_{j \ne i}{[\ehat_{ji}(\gammai, \gammaj) - \ehat_{ij}(\gammai, \gammaj) ] } + 0.5 \sum_{j \ne i}{b_{ij}(\gammai, \gammaj) - b_{ji}(\gammai, \gammaj)}
    \end{equation}
    where $\ehat_{ij}$ is the predicted externality under the true instance $\I$. Next, let $S^q = (\gamma_1^q, \dots, \gamma_n^q)$ denote an equilibrium profile under simplified instance $\tilde{\I}$. This implies for any buyer $i$ unilaterally deviating from $\gammai^q$ to another $\gammai$ (we denote the resulting action profile $S^q_{-i}$) is not beneficial under $\tilde{\I}$: $u_i(S^q; \tilde{I}) \geq u_i(S^q_{-i}; \tilde{I})$. Then by appealing to equation \ref{equation:simple_real_equivalence}, we state that $u_i(S^q_{-i};\I) - u_i(S^q;\I)$, the deviation benefit, is bounded by: $\leq 2|\alpha - 0.5|\sum_{j \ne i}{\max_{\gammai, \gammaj}|\ehat_{ij}(\gammai, \gammaj) - \ehat_{ji}(\gammai, \gammaj)|} + \sum_{j \ne i}{\max_{\gammai, \gammaj}|b_{ij}(\gammai, \gammaj) - b_{ji}(\gammai, \gammaj)|}$. Thus, $S^q$ is an $\varepsilon$ equilibrium for the original instance $\I$ for the states $\varepsilon$ values.
\end{proof}

\begin{cor}\label{cor:potential_best_response}
    For an instance $\I = (g_i, e_{ij}, b_{ij}, \alpha)$, let  instance $\tilde{\I} = (\tilde{g}_i, \tilde{e}_{ij}, \tilde{b}_{ij}, \tilde{\alpha})$, with $\tilde{\alpha} = 0$, $\tilde{g}_i = g_i$, $\tilde{e}_{ij}(\gammai, \gammaj) =$ $ \tfrac{1}{2}\left(e_{ij}(\gammai, \gammaj) + e_{ji}(\gammai, \gammaj)\right)$, and $\tilde{b}_{ij}(\cdot, \cdot) = 0$ for all possible values. Then for $\varepsilon = 2|\alpha - 0.5|\sum_{j \ne i}{\max_{\gammai, \gammaj}|\ehat_{ji}(\gamma_i, \gamma_j) - \ehat_{ij}(\gamma_i, \gamma_j)|} + \sum_{j \ne i}{\max_{\gammai, \gammaj}{|b_{ij}(\gammai, \gammaj) - b_{ji}(\gammai, \gammaj)|}}$, 
    $\tilde{\I}$ has a potential function whose exact PNE strategy is a $\varepsilon$ PNE of $\I$.
\end{cor}

The $\varepsilon$ in this approximate equilibrium is parameterized by the intervention parameter $\alpha$, the instance externalities, and bias $b_{ij}$. With the latter bias component of $\varepsilon$ diminishing linearly to 0 as bias gets closer to 0 or becomes symmetric, we focus on the first term. This linearly approaches 0 in two ways: either $\alpha$ approaches $0.5$, or the externalities become symmetric (i.e. $e_{ij} \approx e_{ji}$)\footnote{While it might seem that it suffices for $\hat{e}_{ij}$ to be symmetric, if the underlying externalities are not, the bias terms will not go to 0}. Overall, we find this result encouraging as it implies the maximum equilibrium deviation linearly reaches 0 either through instance property or $\alpha$ close to 0.5. Further, $\alpha \approx 0.5$ is not extreme but quite reasonable as it charges each buyer half of the predicted net externality they induce. Corollary \ref{cor:potential_best_response} implies that one can obtain this approximate equilibrium by solving a corresponding potential game, which is operationally simple \citep{roughgarden2010algorithmic}. We next study the welfare regret of the equilibrium under our transaction cost and find the results to be positive here as well.

\begin{theorem}\label{theorem:upper_bound_joint}
    For $\varepsilon$ defined in theorem \ref{theorem:joint_eps_pne}, there always exists a $\varepsilon$-PNE with WRaE is at most $\tfrac{n}{2}$.
\end{theorem}
\begin{proof}
     We consider the corresponding simpler version of instance $\I$, $\tilde{\I} = (\tilde{g}_i, \tilde{e}_{ij}, \tilde{b}_{ij}, \tilde{\alpha})$, defined in corollary \ref{cor:potential_best_response}. We note that externality in $\tilde{\I}$ is always symmetric and by corollary \ref{theorem:joint_eps_pne}, a PNE always exists for $\tilde{\I}$, which implies a $\varepsilon$-PNE for instance $\I$, where $\varepsilon$ is equal to the one in the theorem statement. Thus, it suffices to consider the WRaE of instance $\tilde{I}$ for the remainder of the proof.  
     
     Let the social optimal be $S^* = (\gamma_1^*, \dots, \gamma_n^*)$, with each buyer attaining utility $\tilde{u}_i(S^*) = \tilde{g}_i(\gammai^*) - \sum_{j \ne i}{\tilde{e}_{ij}(\gammai^*, \gammaj^*)}$ for a total welfare of $sw(S^*)$. If the social optimal is an equilibrium, then the best-case welfare regret is 0. If not, then by theorem \ref{theorem:joint_eps_pne}, we know that a pure equilibrium of $\tilde{\I}$ can be reached by playing the sequential best response (buyers play the best response one by one). Without loss of generality, suppose buyer 1 is unhappy at the social optimal and changes their decision from $\gamma_1^*$ to $\gamma_1^1$ (denotes that this is buyer 1's first best response), and define this new state as $S_{1,1}$ (denotes buyer $1$ has played their first best response). Thus: 
    \begin{equation}\label{eq:wrae_upper_br}
        \tilde{g}_1(\gamma_1^1) - \sum_{j \ne 1}{\tilde{e}_{1j}(\gamma_1^1, \gamma_j^*)} > \tilde{g}_1(\gamma_1^*) - \sum_{j \ne 1}{\tilde{e}_{1j}(\gamma_1^*, \gamma_j^*)}
    \end{equation}
    Define $\Delta \tilde{g}_i(\gamma_i^t, \gamma_i) = \tilde{g}_i(\gamma_i^t) - \tilde{g}_i(\gammai)$ and $\Delta \tilde{e}_{ij}(\gamma_i^t, \gammai, \gammaj) = \tilde{e}_{ij}(\gammai^t, \gammaj) - \tilde{e}_{ij}(\gammai, \gammaj)$. With this new notation, we can express equation \ref{eq:wrae_upper_br} as $\Delta \tilde{g}_1(\gamma_1^1, \gamma_1^*) > \sum_{j \ne 1}{\Delta \tilde{e}_{1j}(\gamma_1^1, \gamma_1^*, \gammaj^*)}$ and exploiting symmetric externality of $\tilde{I}^s$, the social optimal at $S_{1,1}$ can be succinctly expressed as $sw(S_{1,1}) = sw(S^*) + \Delta \tilde{g}_1(\gamma_1^1, \gamma_1^*) - 2\sum_{j \ne 1}{\Delta \tilde{e}_{1j}(\gamma_1^1, \gamma_1^*, \gammaj^*)}$ (since externalities are symmetric, buyer $1$'s new decision affects externality $\tilde{e}_{1j}$ and $\tilde{e}_{j1}$ equally). This relationship is in fact satisfied between any two consecutive best response steps: let state $S_{i,k}$ denote when buyer $i$ plays her $k^{th}$ best response, which occurs right after state $S_{h,\ell}$ where buyer $h$ plays her $\ell^{th}$ best response (note that not every buyer needs to update at every round). Also, for any buyer $j$, let $t_j$ represent the number of times they have played best response up to time $t$ (i.e. $t_{h} = \ell$). The following two invariants hold: 
    \begin{align}\label{eq:upper_bound_inv_1}
        \Delta \tilde{g}_i(\gammai^k, \gammai^{k-1}) > \sum_{j \ne i}{\Delta \tilde{e}_{ij}(\gammai^k, \gammai^{k-1}, \gammaj^{t_j})}\\
        sw(S_{i,k}) = sw(S_{h,\ell}) + \Delta \tilde{g}_i(\gammai^k, \gammai^{k-1}) - 2\sum_{j \ne i}{\Delta \tilde{e}_{ij}(\gammai^k, \gammai^{k-1}, \gammaj^{t_j})}
    \end{align}
     Let $S^{q} = (\gammai^q, \dots, \gamma_n^q)$ be the pure equilibrium state reached by this best response cycle starting from social optimal. By repeatedly applying the two invariants above to $sw(S^*)$ and expanding, we can relate $sw(S^q)$ to $sw(S^*)$ as follows:
    \begin{equation}\label{eq:upper_bound_eq}
        sw(S^q) = sw(S^*) + \sum_{i=1}^{n}\Delta \tilde{g}_i(\gammai^q, \gammai^*) - 2\sum_{i=1}^{n}\sum_{j \ne i}{\Delta \tilde{e}(\gammai^q, \gammai^*, \gammaj^*)}
    \end{equation}
    Observe that summing the first invariant (equation \ref{eq:upper_bound_inv_1}) across $n$ implies that:
    \begin{equation}
        \sum_{i=1}^{n}\Delta \tilde{g}_i(\gammai^q, \gammai^*) \geq \sum_{i=1}^{n}\sum_{j \ne i}{\Delta \tilde{e}_{ij}(\gammai^q, \gammai^*)}
    \end{equation} 
    Let $\sum_{i}\sum_{j \ne i}{\Delta \tilde{e}_{ij}(\gammai^e, \gammai^*)} \triangleq c$, and note that $c \in [0,n]$ due to the boundedness of $\sum_{ij}{\tilde{e}_{ij}}$. Thus we have: $\sum_{i=1}^{n}\Delta \tilde{g}_i(\gammai^e, \gammai^*) \geq c \implies sw(S^*) \leq n-c$ since maximum net gain for any buyer is $1$. Thus, $\text{WRaE} \leq n-c$. Next, observe by equation \ref{eq:upper_bound_eq}, $sw(S^q) \geq sw(S^*) + c - 2c = sw(S^*) - c \implies \text{WRaE} \leq c$. Combining these, we have that $\text{WRaE} = \min(c, n-c) \leq \tfrac{n}{2}$.    
\end{proof}

To summarize these results, joint externality is a richer model where the picture without any intervention is even more grim. Theorem \ref{thm:joint_noeq}, shows the existence of instances where even $\varepsilon$-PNE does not exist for $\varepsilon < 1$, the maximum possible utility. Theorem \ref{thm:joint_noeq_wrae} shows even if PNE exists, the welfare regret is maximal. These reinforce the idea that equilibrium properties without intervention are poor. We next consider introducing our parameterized transaction cost to this setting which leads to significant improvements. We show in theorem \ref{theorem:joint_eps_pne} that a $\varepsilon$-PNE always exists, where $\varepsilon \rightarrow 0$ linearly as $\Oh(b) \rightarrow 0$ and $\alpha \rightarrow 0.5$. We also show in theorem \ref{theorem:upper_bound_joint} that the quality of the equilibrium achieved here is significantly better, with welfare regret less than $n/2$. Note that this WRaE roughly matches the tight WRaE bound for the independent externality setting (theorem \ref{theorem:wrae_offline}) with $\alpha=0.5$. This suggests that $\alpha$ close to $0.5$ is a silver bullet: it imposes a reasonable cost, charging each buyer half of the net externality they induce, while having favorable equilibrium properties regardless of the externality model. 

In Appendix \ref{app:Experiments} we experimentally validate our findings on the positive impact of our proposed transaction cost across both externality models. On a dataset derived from AWS Data Exchange, we show significant increase in social welfare with increasing $\alpha$ for the standard externality model (Figure 1); crucially even a small $\alpha$ leads to a meaningful increase, with nearly maximal improvements achieved already at $\alpha=0.6$. In the joint externality setting (Figure 2), we see social welfare improvements concentrated at $\alpha=0.5$ as expected, but meaningful increase at other values as well. These results show that a range of $\alpha$ may lead to improved welfare, a significant operational advantage.

\section{Discussion}\label{sec:discussion}
Despite recent research on new market structures for data, real-world data markets remain far simpler: sellers post fixed prices and buyers are unfettered in their purchases. Our work fills a gap in the literature by studying this simple market under the unique characteristics of data: imperfect valuations, free replicability, and negative externality. The presence of externality allows us to naturally model buyer interactions as a simultaneous game. While this game has poor equilibrium properties by itself, a simple transaction cost can greatly improve these characteristics. For a standard externality model, our intervention is nearly perfect, guaranteeing a dominant strategy for buyers and leading to near-optimal social welfare. This intervention also fares well in the more realistic setting where buyers learn valuations through repeated interaction. We prove that buyers can learn to play their dominant strategy, while still achieving low social welfare regret. 
Lastly, we analyze an extension of this model by considering a richer class of externality. Although the equilibrium guarantees we can provide in this richer setting are naturally weaker, our proposed transaction cost significantly improves upon the status quo here as well.

Our work illustrates that when coupled with simple interventions, fixed-price data markets can be an elegant solution for a challenging product. It also leaves open several intriguing questions. We analyze a single transaction cost; it is unclear whether this is optimal or even what the space of such interventions is. Optimizing for non-utilitarian notions of welfare, like egalitarianism or Nash welfare is also an interesting research direction \citep{moulin2004fair}. Our intervention leads to platforms learning or eliciting buyer externalities since transaction cost is a simple sum of these quantities. Designing truthful mechanisms or prediction algorithms for this is an important research question \cite{chen2020truthful, kong2020information, zohar2008mechanisms}. Similarly, understanding the long-term strategic perspective of sellers and their externality within this game will be insightful. Lastly, extending the joint externality model to the online setting with buyers playing based on learned valuations remains an interesting open problem. This is spiritually similar to work on learning and regret minimization in repeated games \citep{cesa2006prediction, slivkins2019introduction}.

\bibliography{arxiv/bibliography}

\newpage
\appendix
\section{Appendix}
\subsection{Discussion on externality models}\label{appendix:externality}

The externality model we primarily focus on, which we also denote as independent externality, is the model considered in related literature \citep{aseff2008optimal, li2019facility, agarwal2020towards}. While simple, we think it is an appropriate and important model for many real-world settings. We give two different examples. Consider two firms acquiring data to improve their in-house models. The two firms are competitors and each firm’s utility depends on the relative performance of their model with respect to their competitor. In this case, $g_i(\gamma_i)$ can represent the accuracy of firm $i$’s model due to data purchases $\gamma_i$, and $e_{ij}$ can simply be buyer $i$'s approximation for the improvement achieved by their competitor $j$ due to their data purchase $\gammaj$: $\hat{g}_j(\gamma_j)$. Note this externality is independent of $i$'s decision, and firm $i$’s utility (without transaction cost) is then $u_i(\gamma_i, \gamma_j) = g_i(\gamma_i) - \hat{g}_j(\gamma_j)$. This is called additively separable externalities in \citep{ agarwal2020towards}, and we note that this is a common scenario when firms purchase data. 

This model can also capture indirect competition. Consider two trading firms $C_a$ and $C_b$ that trade different financial products. $C_a$ focuses on equities and $C_b$ focuses on Forex. They acquire data to improve their respective performance, and are targeting different data sources since they operate on different markets and products. However, they are competitors with respect to recruiting talent and securing new investments. If $C_b$ purchases data it deems very prescient, it still negatively affects $C_a$, since a more profitable $C_b$ can better compete in the talent and investment pool it shares with $C_a$. Thus $C_b$’s action induces negative externality on $C_a$, independent of the $C_a$'s data purchase decision (the gain from which is encapsulated in $g_i$). 
        
The joint externality model is indeed more general and can capture the above examples. In these examples, the independent model can be seen as capturing indirect competition. However, it is possible that buyers might be competitors (and exert externality) purely with respect to data. Consider two firms in the same consumer industry, with consumer preferences or market data being sold. If both firm purchase the same data, they target the same group of consumers and eat into each other's profit. If the data they purchase don't overlap, they each target their respective groups and stay out of each other's way. In this setting, externality suffered by buyer $i$ depends on both of their actions.

\section{Section \ref{section:learned_valuations} Proofs}\label{app:learn}

\subsection{Proof of Lemma \ref{lemma:concentration}}
\begin{proof}
    Consider a buyer $i$ and a seller set $\gamma$. If it is not in the active set or has never been selected, the lemma trivially holds since $n_i^t(\gamma) = 0$, implying $c_t^i(\gamma) > 2$. If $\gamma$ has been selected at least once, denote the $\ell^{th}$ realizations by $y^{\ell}_{i}$, which consists of the sampled net gain and externality-based transaction cost. For the purpose of this analysis, we care only about the following realized parameters: $y^{\ell}_{i} = (g^{\ell}_i(\gamma), \ehat^{\ell}_{1i}(\gamma), \ehat^{\ell}_{2i}(\gamma), \dots, \ehat^{\ell}_{ni}(\gamma))$. Next, we can write the empirical mean of the effective utility from the first $h$ times $\gamma$ has been selected by buyer $i$ as $\udihat(\gamma; h) = \udihat(y^1_i, \dots, y_i^h) = \tfrac{1}{h}\sum_{\ell=1}^{h}\left({g^\ell_i(\gamma) - \ab\sum_{j \ne i}{\ehat^{\ell}_{ji}(\gamma)}}\right)$. Observe that $\udihat(y^1_i, \dots, y_i^h)$ satisfies the bounded difference property: $\forall \ell \in [1, \dots, h], \forall y^{\ell}_i, \, \sup_{y'^{\ell}_i \in [0,1]^n}{|\udihat(\dots, y^{\ell}_i, \dots) - \udihat(\dots, y'^{\ell}_i, \dots)|} \leq \tfrac{2}{h}$. This allows us to apply McDiarmid's inequality on the function $\udihat$ over the random samples our $h$ selections of $\gamma$ results in, denoted by $Y^1_i, \dots, Y^{h}_i$. We note that $\E[\udihat(Y_1, \dots, Y_h)] = g_i(\gamma) - \ab \sum_{j \ne i}{\ehat_{ji}(\gamma)} = \udi(\gamma)$. Thus, by McDiarmid's inequality, we have: 
    \begin{equation}
        \forall i, \forall \gamma, \forall h \,\, \mathbb{P}\left[|\udihat({\gamma}, h) - \udi(\gamma)| \leq \sqrt{\tfrac{12\log(T)}{h + 1}}\right] \geq 1 - \tfrac{2}{T^6} 
    \end{equation}
    Thus for all $h$, where $h$ is the number of times buyer $i$ selects $\gamma$, we have the bad event (no concentration) probability is $\tfrac{2}{T^6}$. Since $h \leq t$ and $n_t^i(\gamma) \leq t$, we can apply a union bound over all possible $h$ and arrive at:
    \begin{equation}
        \forall i, \forall \gamma, \forall t \,\, \mathbb{P}\left[|\udihat(\gamma, n_t^i(\gamma)) - \udi(\gamma)| \leq \sqrt{\tfrac{12\log(t)}{n_t^i(\gamma) + 1}}\right] \geq 1 - \tfrac{2}{T^5} 
    \end{equation}
    By applying a union bound over all $t \in [1, \dots, T]$, we have that an event defined by $\mathcal{E}_i(\gamma) = \left\{ \forall t \, |\udihat(\gamma, n_t^i(\gamma)) - \udi(\gamma)| \leq \sqrt{\tfrac{12\log(t)}{n_t^i(\gamma) + 1}} \right\}$ holds with probability greater than $1 - \tfrac{2}{T^4}$.
    We would like to now show that for each buyer, this holds uniformly for all possible seller sets, and not just for a fixed $\gamma$. Since the total number of seller sets is exponential, a naive application of union bound provides a poor bound. However, since any inactive seller set trivially satisfies the bound with probability 1, it suffices to consider only the active seller sets for a buyer $i$, $\mathcal{A}_i$. While this set's size is bounded since we add at most one arm every round, it is random in its composition. For $j \in [1, \dots, t]$, let $Z_{i}^j$ denote the $j^{th}$ arm activated by buyer $i$. $Z_i^{j}$ is a random variable and $\{Z_i^{1}, \dots, Z_i^{t}\}$ is the set of all activated arms\footnote{If the number of activated arms is less than $t$, then for $j > |\mathcal{A}_i|$, let $Z_{i}^{j}$ be the last arm activated}. For a seller set $\gamma$, note that the event $\{Z_i^{j} = \gamma \}$ depends on the outcome of \emph{previously activated sets}, whereas $\mathcal{E}_i(Z_i^{j})$ is purely based on the observations derived from the seller set $Z_i^{j}$, whatever that happens to be. In other words, the event $\{Z_i^{j} = \gamma\}$ is independent of $\mathcal{E}_i(Z_i^{j})$. Thus we have that $\forall i, \forall Z_i^{j}$, the clean event holds for the $Z_i^{j}$ activated arm with the following probability:
    \begin{equation}
        \Pb\left[\mathcal{E}_i(Z_i^{j})\right] = \sum_{\gamma}{\Pb[\mathcal{E}_i( Z_i^{j})|Z_i^{j}=\gamma]\Pb[Z_i^{j}=\gamma]} = \sum_{\gamma}{\Pb[\mathcal{E}_i(\gamma)]\Pb[Z_i^{j}=\gamma]} \leq 1 - \tfrac{2}{T^4}
    \end{equation}
    where we note that $\Pb[\mathcal{E}_i(\gamma)]$ is a constant and can be taken outside the sum. Thus, we have a concentration result for each active seller set. Now we apply a union bound over the whole active set $\mathcal{A}_i$. Noting that $|\mathcal{A}_i| \leq T$, we have that: $\forall i, \mathbb{P}\left[\forall \, \gamma \in \mathcal{A}_i \,, \mathcal{E}_i(\gamma)\right]  \geq 1 - \tfrac{2}{T^3}$. Lastly, note that the event $\mathcal{E} = \bigwedge_{i=1}^{n}{\mathcal{E}_i}$, and thus we apply a union bound over all the buyers. Assuming that the number of buyers is smaller than $T$, we arrive at the statement of the lemma.
\end{proof}

\subsection{Proof of Lemma \ref{lemma:slivkins_lemma}}
\begin{proof}
    Fix any buyer $i$, any seller set $\gamma$, and a time $t$. If $\gamma$ is not in the active set or never selected, then this claim holds trivially since $c_t^{i}(\gamma) > 2$, allowing us to consider only active arms. Suppose seller set $\gamma$ was last chosen at some time $s \leq t$. Now consider the optimal dominant strategy for buyer $i$, $\gammai^d$, and the two possibilities at time $s$: (1) Either $\gammai^d$ is already part of the active set, or (2) $\gammai^d$ is covered by some other set $\gammai' \in \mathcal{A}_i$ which has confidence radius $c^i_s(\gammai') \geq \tfrac{1}{k}$ (the closest element to $\gammai'$ has to be at least $\tfrac{1}{k}$ away). Then the following holds at time $s \leq t$ for each case:
    \begin{equation}
        \begin{split}
            \text{if (1): } \text{UCB}_i(\gamma) \geq \text{UCB}_i(\gammai^d) = \udihat(\gammai^d, n_t^i(s)) + 2c^i_s(\gammai^d) \geq \udi(\gammai^d)\\
            \text{if (2): } \text{UCB}_i(\gamma) \geq \text{UCB}_i(\gammai') = \udihat(\gammai', n_t^i(s)) + 2c_{s}^{i}(\gammai') \geq  \udi(\gammai') + c_{s}^{i}(\gammai') \geq \udi(\gammai^d)
        \end{split}
    \end{equation}
    where the last inequality in (1) follows from lemma \ref{lemma:concentration} and the last inequality in (2) follows from the fact that $\gammai^d$ is covered by $\gammai'$ (i.e. $D_h(\gammai', \gammai^d) \leq c^i_t(\gammai')$) and thus their difference in utility is equivalently bounded by the closeness property. The following upper-bound for $\text{UCB}_i(\gamma)$ holds regardless:
    \begin{equation}
         \text{UCB}_i(\gamma) = \udihat(\gamma, n_t^i(s)) + 2c^{i}_{s}(\gamma) \leq \udi(\gamma) + 3c^{i}_{s}(\gamma) = \udi(\gamma) + 3\cit(\gamma)
    \end{equation}
    where we use that fact that $c_{s}^{i}(\gamma) = \cit(\gamma)$ since $s$ is the last time $\gamma$ was selected. Putting the upper and lower bounds on $\text{UCB}_i$ together, we have:
    \begin{equation}
         \udi(\gamma) + 3\cit(\gamma) \geq \text{UCB}_i(\gamma) \geq \text{UCB}_i(\gamma') \text{ or }\text{UCB}_i(\gamma^*) \geq \udi(\gammai^*) \implies \Delta_i(\gamma) \leq 3\cit(\gamma)
    \end{equation}
    
    We now move to the second part of the lemma. First note that by property of the Hamming space, any two active seller sets (in fact any two seller sets) must be at least $\tfrac{1}{k}$ apart. Consider two active choices $\gamma^1$ and $\gamma^2$ for buyer $i$, and suppose $\gamma^1$ was activated (at time step $t_1$) before $\gamma^2$ (activated at time-step $t_2$). The fact that $\gamma^2$ was activated implies that it was not covered by $\gamma^1$'s confidence radius at $t_2$, and thus $D_h(\gamma^1, \gamma^2) > c_{t_2}^i(\gamma^1) \geq \tfrac{\Delta_i(\gamma^1)}{3}$. If $\gamma^2$ was activated before $\gamma^1$, we get the opposite result. Combining the two, we have the $D_h(\gamma^1, \gamma^2) \geq \tfrac{1}{3}\min(\Delta(\gamma^1), \Delta(\gamma^2))$. Putting this and the Hamming observation of two elements being at least $\tfrac{1}{k}$ apart, we have the desired result.
\end{proof}

\newpage
\section{Experimental Results}\label{app:Experiments}
We experimentally validate our theoretical insights on the effectiveness of transaction costs. While no publicly available dataset exists for data markets with prices and utilities, we take inspiration from AWS marketplace (a fixed price data market platform) we design a suite of synthetic experiments. We take the 10 different data categories from AWS, and instantiate several data sellers for each category (177 total sellers). For each category, we have several buyers (57 total), where each buyer has zero gain/externality for the sellers \emph{not} in their category, with gains and externalities for their category sampled uniformly. They may purchase from up to 10\% of the sellers in their category. This is meant to model the pricing and budget constraints of buyers. 

\begin{figure}[htb]\label{fig:standard_externality}
\centering
\begin{minipage}{0.8\textwidth}
\centering
\begin{tikzpicture}
  \centering
  \node (img)  {\includegraphics[scale=0.5]{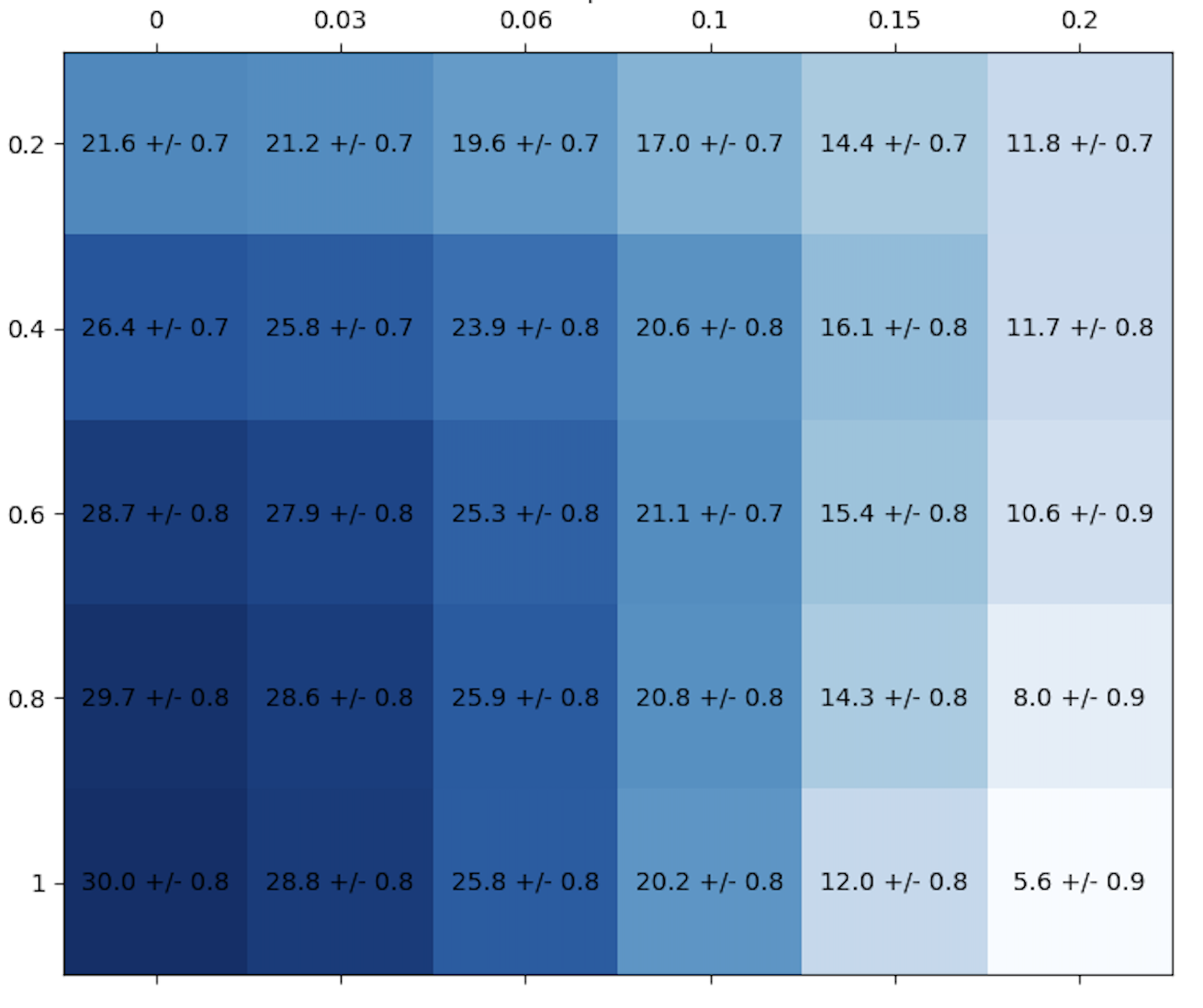}};
  \node[above=of img, node distance=0cm, yshift=-1cm] {\large{Bias}};
  \node[left=of img, node distance=0cm, rotate=90, anchor=center,yshift=-0.7cm,font=] {\large{Alpha}};
 \end{tikzpicture}
 \vspace{-1em}
\caption{Standard Externality Model: Avg \% increase in social welfare (with 90\% confidence interval) from gain maximizing decision.}
\end{minipage}%
\end{figure}

We first consider the standard externality model, with results presented in Figure 1. We use $\alpha=0$ (no intervention) as the baseline which corresponds to buyers picking seller sets with the highest possible gain. We measure the increase in social welfare from this baseline to the social welfare at equilibrium under varying $\alpha$ and bias parameters (denoted by epsilon). As expected, our plots show that increasing $\alpha$ leads to increased social welfare. Also as expected, as the bias of the platform’s externality estimate increases, welfare decreases, with the effect more pronounced as $\alpha$ increases. This is also reasonable since enforcing large transaction costs with very inaccurate predictions is inadvisable. Fortunately, the results show that $\alpha=0.6$ can capture most of the increase in welfare, without suffering too much when then the externality estimates by the platform are inaccurate.

\begin{figure}[htb]\label{fig:joint_externality}
\centering
\begin{minipage}{0.8\textwidth}
\centering
\begin{tikzpicture}
  \centering
  \node (img)  {\includegraphics[scale=0.50]{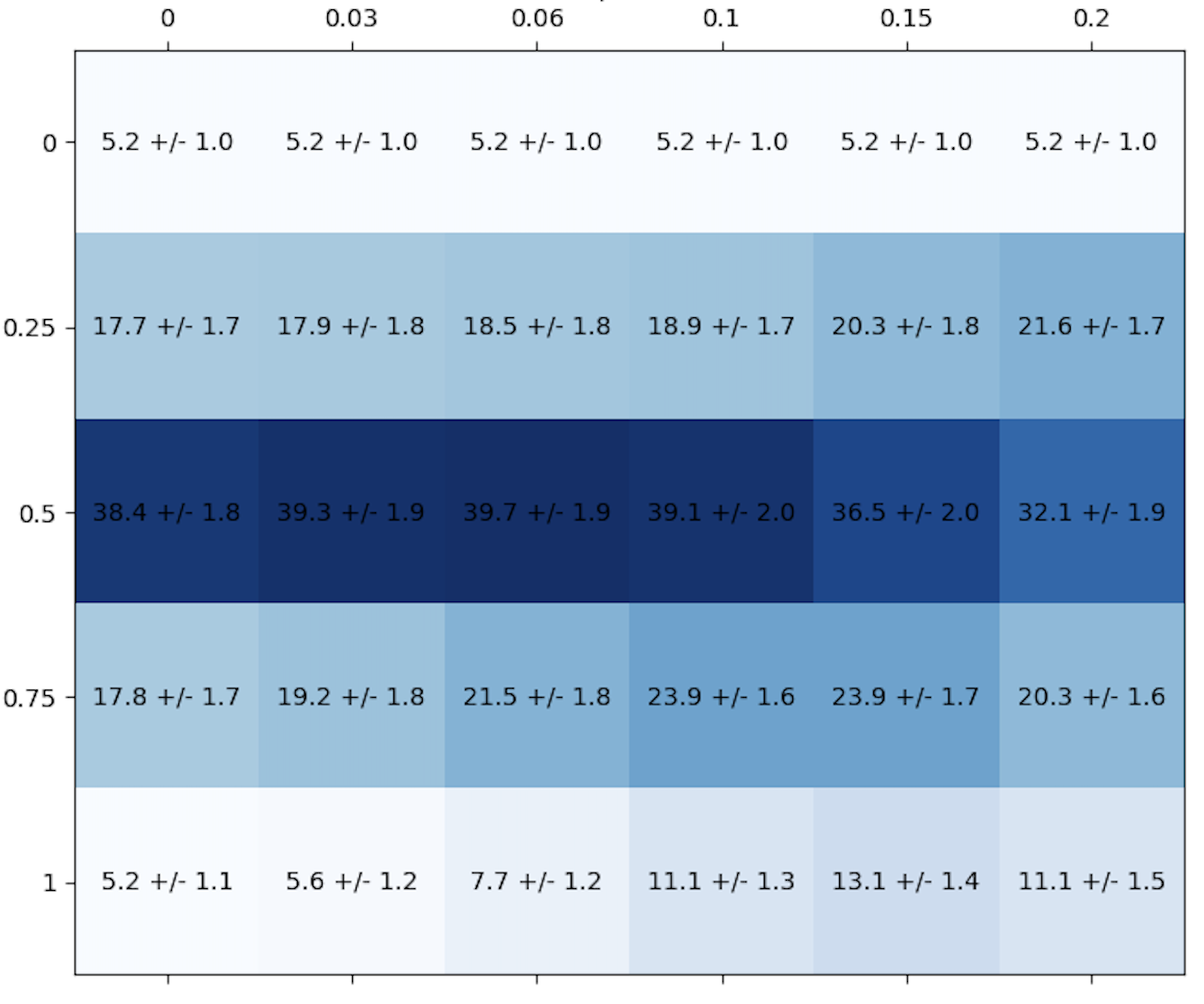}};
  \node[above=of img, node distance=0cm, yshift=-1cm] {\large{Bias}};
  \node[left=of img, node distance=0cm, rotate=90, anchor=center,yshift=-0.7cm,font=] {\large{Alpha}};
 \end{tikzpicture}
 \vspace{-1em}
\caption{Joint Externality Model: Avg \% increase in social welfare (with 90\% confidence interval) from gain maximizing decision.}
\end{minipage}%
\end{figure}

For the joint externality setting, whose results are presented in Figure 2, there is no dominant strategy, even without any intervention. As such, we use the gain maximizing choice from before as the baseline (which isn’t necessarily an equilibrium). Our theoretical results show with intervention the resulting game is an approximate potential game with an approximate equilibrium. As such, we use a best response algorithm to find the approximate equilibrium. Once again, we plot the increase in social welfare (sum of all buyer utilities) from the baseline for different values of $\alpha$ and bias. As expected from our theoretical results, setting $\alpha$ closer to 0.5 from either direction leads to a higher increase in social welfare. The results are quite symmetric. We also notice that the effect of bias is less clear here than in the standard case.

Operationally, these results corroborate our assertion that $\alpha$ close to 0.5 is the ideal parameter. In the standard setting, it offers much of the welfare benefits of setting it higher without any of the drawbacks due to bias, and in the joint setting, it is outright the best parameter.

\end{document}